\documentclass[11pt]{article}

\usepackage{empheq}
\usepackage{amsmath,amsfonts,amssymb,amsbsy,amscd}
\usepackage{mathrsfs,array}
\usepackage{subcaption}
\usepackage[a4paper, total={6in, 8in}]{geometry}
\usepackage{amsthm}
\usepackage[numbers,sort&compress]{natbib}
\newtheorem{theorem}{Theorem}[section]%  meant for continuous numbers
\newtheorem{lemma}[theorem]{Lemma}
\newtheorem{proposition}[theorem]{Proposition}% 
\newtheorem{remark}[theorem]{Remark}%

\usepackage{tikz}
\usetikzlibrary{arrows.meta}
\usetikzlibrary{3d}
\usetikzlibrary{shadings}
\usetikzlibrary{math}
\usepackage{adjustbox}\usepackage{multirow}
\usepackage{upgreek}
\usetikzlibrary{intersections, arrows, automata, backgrounds, calendar, chains, matrix, mindmap, patterns, petri, shadows, shapes.geometric, shapes.misc, spy, trees, shapes}

\definecolor{darkgreen}{rgb}{0, 0.7, 0}
\definecolor{orange}{rgb}{0.98, 0.6, 0.01}
 	\definecolor{napiergreen}{rgb}{0.16, 0.5, 0.0}

\usepackage{upgreek}

\newcolumntype{R}[2]{%
    >{\adjustbox{angle=#1,lap=\width-(#2)}\bgroup}%
    l%
    <{\egroup}%
}
\makeatletter
\newcommand{\thickhline}{%
    \noalign {\ifnum 0=`}\fi \hrule height 1pt
    \futurelet \reserved@a \@xhline
}
\newcolumntype{"}{@{\hskip\tabcolsep\vrule width 1pt\hskip\tabcolsep}}
\makeatother

\usepackage{pifont}% http://ctan.org/pkg/pifont
\newcommand{\cmark}{\ding{51}}%
\newcommand{\xmark}{\ding{55}}%

\usepackage{hyperref}
\hypersetup{
    colorlinks=true,
    linkcolor=blue,
    urlcolor=red,
    linktoc=all,
    allcolors=blue
}

\usepackage[nameinlink,capitalise,noabbrev]{cleveref}
\makeatletter
\newcommand{\myref}[1]{\cref{#1}\mynameref{#1}{\csname r@#1\endcsname}}
\newcommand{\Myref}[1]{\Cref{#1}\mynameref{#1}{\csname r@#1\endcsname}}

 	\definecolor{burgundy}{rgb}{0.5, 0.0, 0.13}
 	\definecolor{napiergreen}{rgb}{0.16, 0.5, 0.0}
    \definecolor{lasallegreen}{rgb}{0.03, 0.47, 0.19}
    \definecolor{airforceblue}{rgb}{0.36, 0.54, 0.66}

\definecolor{mycolor}{rgb}{1,0.2,0.3}
\definecolor{beaublue}{rgb}{0.74, 0.83, 0.9}
\definecolor{oceanboatblue}{rgb}{0.0, 0.47, 0.75}
\definecolor{lightblue}{rgb}{0.78, 0.95, 1.0}	
\definecolor{amber}{rgb}{1.0, 0.6, 0.1}
\definecolor{darkblue}{rgb}{0.0, 0.0, 0.55}
\definecolor{lightgreen2}{rgb}{0.898, 1.000, 0.835}
\definecolor{lightgray}{rgb}{0.9, 0.9, 0.9}
\definecolor{lightblue}{rgb}{0.80, 0.93, 0.95}
\definecolor{orangepeel}{rgb}{1.0, 0.85, 0.5}

\definecolor{amber3}{rgb}{1.0, 0.74, 0.5}
\definecolor{lightgreen3}{rgb}{0.949, 1.0, 0.918}
\definecolor{lightblue3}{rgb}{0.90, 0.97, 0.975}
\definecolor{orangepeel3}{rgb}{1.0, 0.93, 0.75}

\definecolor{azure}{rgb}{0.0, 0.5, 1.0}

\usepackage{accents}
\newlength{\dhatheight}

\newcommand{\doublecheck}[1]{%
    \settoheight{\dhatheight}{\ensuremath{\check{#1}}}%
    \addtolength{\dhatheight}{-0.10ex}%
    \check{\vphantom{\rule{1pt}{\dhatheight}}%
    \smash{\check{#1}}}}

\numberwithin{equation}{section}

\allowdisplaybreaks

\title{Compressible N-phase fluid mixture models}

\author{M.F.P. ten Eikelder$^\dag$\thanks{e-mail: \texttt{marco.eikelder@tu-darmstadt.de}}, E.H. van Brummelen$^\ddag$, D. Schillinger$^\dag$
}
\date{%
    $^\dag$Institute for Mechanics, Computational Mechanics Group, Technical University of Darmstadt\\
    $^\ddag$Department of Mechanical Engineering, Eindhoven University of Technology,\\ 5600MB Eindhoven, The Netherlands
}

\newcommand{\nn}{\nonumber}

\newcommand{\br}{\mathbf{r}}
\newcommand{\bl}{\mathbf{l}}
\newcommand{\mA}{\alpha}
\newcommand{\mB}{{\beta}}
\newcommand{\mC}{{\gamma}}

\newcommand{\mT}{{\tau}}

\newcommand{\bu}{\mathbf{u}}

\newcommand{\bw}{\mathbf{w}}

\newcommand{\bv}{\mathbf{v}}

\newcommand{\bx}{\mathbf{x}}

\newcommand{\bJ}{\mathbf{J}}
\newcommand{\bh}{\mathbf{h}}
\newcommand{\bH}{\mathbf{H}}
\newcommand{\bj}{\mathbf{j}}

\newcommand{\bchi}{\boldsymbol{\chi}}
\newcommand{\bnu}{\boldsymbol{\nu}}
\newcommand{\trho}{\tilde{\rho}}
\newcommand{\bpi}{\boldsymbol{\pi}}

\def\be{\begin{equation}}
\def\ee{\end{equation}}
\def\ba{\begin{array}}
\def\ea{\end{array}}
\def\bea{\begin{eqnarray}}
\def\eea{\end{eqnarray}}
\def\beas{\begin{eqnarray*}}
\def\eeas{\end{eqnarray*}}
\newcommand{\bseq}{\begin{subequations}}
\newcommand{\eseq}{\end{subequations}}

\begin{document}

\maketitle

% REQUIRED
\begin{abstract}
Fluid mixture models are essential for describing a wide range of physical phenomena, including wave dynamics and spinodal decomposition. However, there is a lack of consensus in the modeling of compressible mixtures, with limited connections between different classes of models. On the one hand, existing compressible two-phase flow models accurately describe wave dynamics, but do not incorporate phase separation mechanisms. On the other hand, phase-field technology in fluid dynamics consists of models incorporating spinodal decomposition, however, a general phase-field theory for compressible mixtures remains largely undeveloped.

In this paper, we take an initial step toward bridging the gap between compressible two-phase flow models and phase-field models by developing a theory for compressible, isothermal N-phase mixtures. Our theory establishes a system of reduced complexity by formulating N mass balance laws alongside a single momentum balance law, thereby naturally extending the Navier-Stokes Korteweg model to N-phases and providing the Navier-Stokes Cahn-Hilliard/Allen-Cahn model for compressible mixtures. Key aspects of the framework include its grounding in continuum mixture theory and its preservation of thermodynamic consistency despite its reduced complexity.
\end{abstract}

% REQUIRED
%\begin{keywords}

%\end{keywords}


\section{Introduction}\setcounter{footnote}{2}
\renewcommand{\thefootnote}{\fnsymbol{footnote}}
\subsection{Background}
The modeling of compressible multiphase fluid mixtures is fundamental to a wide range of scientific and engineering applications, including geophysical flows, aerospace engineering, and industrial processes. A key challenge in describing such systems lies in formulating a mathematical framework that consistently captures the interactions between phases\footnote{We utilize ``phase" to denote different fluid constituents (e.g. air and water).} while accounting for compressibility effects, diffusion, and interfacial dynamics. A widely used class of PDE-based models for multiphase flows consists of sharp-interface models, which either explicitly track the phase boundaries and prescribe conditions at the interface \citep{ghias2007sharp,shen20203d}, or are approximated by smooth interface (level-set) approximations \citep{sussman1994level,sethian1999level,ten2021novel}. These models have been successfully applied to problems where the interface between phases remains well-defined. Complex mixture problems, however, where phases are dispersed throughout the domain and interpenetrate one another, or emerge from phase separation, do not belong to the class of sharp-interface problems. As a consequence, these problems demand an alternative modeling approach. Below we discuss two classes of models that describe these complex mixtures: \textbf{compressible two-phase models} and \textbf{phase-field models}.

\textbf{Compressible two-phase flow models} play a fundamental role in describing mixture flows, where compressibility and interfacial dynamics are essential. Among these models, the Baer-Nunziato model is often regarded as one of the most comprehensive frameworks for describing two-phase flows \citep{baer1986two}. This diffuse-interface, non-equilibrium model consists of separate mass, momentum, and energy balance equations for each phase, along with an additional equation governing the topology of the two-fluid interface. Over time, significant extensions have been proposed, including generalizations to $N$-component mixtures (see e.g. \cite{muller2016closure}), broadening the applicability to more complex multiphase systems.

Despite their generality, Baer-Nunziato-type models introduce considerable computational complexity due to their large number of equations and wave families. To address this, reduced models with fewer equations have been developed \citep{kapila2001two,murrone2005five}. These reduced formulations are often derived through relaxation procedures that assume velocity, pressure, or temperature equilibrium between phases \citep{saurel1999multiphase}. However, such simplifications come at a cost: reduced models may fail to strictly conserve energy, suffer from an ill-posed mixture speed of sound \citep{saurel2009simple}, or lose asymptotic consistency with the full nonlinear system \citep{saleh2020some}. These limitations motivate the exploration of alternative modeling approaches for compressible mixtures, particularly those that incorporate interfacial dynamics and phase transition effects in a thermodynamically consistent manner.

Over the past decades, \textbf{phase-field modeling} has become a valuable tool for moving boundary problems in computational fluid mechanics \citep{anderson1998diffuse,gomez2018computational}. It addresses two key challenges: geometrical representation and physical modeling. The phase-field variable naturally captures topological changes, providing advantages over traditional interface-tracking methods for complex flows. Additionally, phase-field models offer a thermodynamically consistent framework for capillary effects, phase transitions, and material composition. They also provide an effective description of wetting, making them well-suited for complex wetting phenomena \citep{aland2021unified,bhopalam2022elasto,demont}. %We discuss two phase-field theories that are essential in the context of 
%van2021adaptive,
%Over the last decades phase-field modeling has become a mature technology for a wide range of problems in fluid dynamics, such as elasto-capillarity, wetting, jetting, . 

The state-of-the art phase-field models for fluids may roughly be classified into two categories. The first category describes describes single-fluid flow with phase change \citep{bresch2019navier,gomez2010isogeometric} through the Navier-Stokes Korteweg (NSK) model. The two phases correspond to the liquid and gaseous states of the same component. In the NSK model, the density field itself acts as phase-field variable. Areas with low density identify as vapor and high-density regions as liquid. Liquid-vapor transitions are often incorporated in the NSK model by means of the Van der Waals capillarity model, named after the 1910 Nobel Laureate in physics \citep{waals1894}. % This equation of state is applicable in several physical states, even including solid phases \cite{james1980propagation}. As such, this model can be adopted in a wide range of problems, including the simulation of cavitation \cite{hu2023direct}, nucleation and film boiling \cite{liu2015liquid}, or phase-transition-driven implosion of thin structures \cite{bueno2015interaction}.

%The second category consists of models for single-phase multi-components flows. 
The second category consists of models that describes viscous (incompressible) mixture flow through the Navier-Stokes Cahn-Hilliard (Allen-Cahn) (NSCHAC) model. The origin of NSCHAC models may be traced back to 1977 when Hohenberg and Halperin proposed a NSCHAC model via the coupling between the Navier-Stokes equations, describing viscous fluid flow, and the Cahn-Hilliard equation, describing spinoidal decomposition \citep{hohenberg1977theory}. The most important limitation of this model is that the density of the fluids are assumed equal. This precludes the applicability of the model to a large number of practical problems such as ink-droplet dynamics. Since the end of the last century, a number of efforts of extending the model to the case of non-matching densities have been made, e.g. \cite{lowengrub1998quasi,boyer2002theoretical,ding2007diffuse,abels2012thermodynamically,shen2013mass,aki2014quasi}. %,shokrpour2018diffuse
Each of these works aims to describe the same physics, yet the proposed models are different. The occurrence of a many NSCH(AC) models shows that consensus in the realm of single-phase multi-component flows has been missing. Note that this is in contrast to single-fluid flow with phase change, where the NSK model is the widely adopted model. In two recent articles we have proposed a unified mixture-theory framework of NSCHAC models as a resolution to this open problem (see \cite{eikelder2023unified,eikelder2024unified} for the theoretical framework and \cite{ten2024divergence} for supporting benchmark computations). The framework indicates a single NSCHAC model that is invariant to the choice of fundamental variables.

\subsection{Objective and main results}
Although significant progress has been made in modeling multiphase flows, a general phase-field theory for compressible mixtures remains largely undeveloped, despite some important contributions (see e.g. \cite{malek2008thermodynamic,mukherjee2024mixtures} and the references therein). While compressible two-phase flow models exist and provide a well-established framework for describing compressible mixtures, they differ fundamentally from phase-field models in both structure and formulation. Phase-field models, on the other hand, have been extensively used for capturing interfacial dynamics in incompressible fluid mixtures but have not been systematically extended to fully compressible regimes. As a result, the connection between compressible two-phase flow models and phase-field approaches is largely unresolved. Although both modeling strategies aim to describe similar physical processes, their underlying formulations remain distinct. This disconnect has hindered the development of a unified theoretical framework that accommodates both compressibility effects and the diffuse-interface nature of phase-field models.

In this paper, we take an initial step toward expanding the theoretical foundation of compressible mixture modeling by developing a thermodynamically consistent theory for compressible, isothermal $N$-phase mixtures. As part of this development, we draw connections to both phase-field models and compressible two-phase flow models, outlining key points of overlap and distinction. Our approach introduces a set of $N$ mass balance equations coupled with a single momentum balance law.
%By integrating elements from both compressible two-phase flow models and phase-field formulations, the proposed theory establishes a foundation for modeling a broad class of multiphase flows where compressibility, phase separation, and interfacial evolution are all critical. This work represents an essential step toward developing a generalized mixture framework that unifies existing methodologies and extends their applicability to previously intractable problems in microfluidics, phase transitions, and complex multiphase systems. 
In particular, we derive the following multi-phase-field model for constituents $\mA = 1, ..., N$:
\begin{subequations}\label{eq: model intro}
  \begin{align}
   \partial_t (\rho \bv) + {\rm div} \left( \rho \bv\otimes \bv \right) + \sum_{\mB} \trho_\mB\nabla\mu_\mB %& \nn\\
    - {\rm div} \left(   \nu (2\mathbf{D}+\lambda({\rm div}\bv) \mathbf{I}) \right)-\rho\mathbf{g} &=~ 0, \label{eq: model intro: mom}\\
  \partial_t \trho_\mA  + {\rm div}(\trho_\mA  \bv) - {\rm div} \left(\sum_{\mB} \mathbf{M}_{\mA\mB}\nabla\mu_\mB\right)  + \sum_{\mB} m_{\mA\mB} \mu_\mB &=~0,\label{eq: model intro: mass}\\
  \mu_\mA - \dfrac{\partial \Psi}{\partial \trho_\mA} + {\rm div} \left(  \dfrac{\partial \Psi}{\partial \nabla \trho_\mA} \right)&=~0,\label{eq: model intro: chem}
  \end{align}
\end{subequations}
where $\mathbf{v}$ is the fluid velocity, $\trho_\mA$ are the partial densities, $\rho= \sum_\mB\trho_\mB$ is the mixture density, $\mathbf{g}$ is the force vector, $\nu$ is the dynamical viscosity, $\nu \lambda$ is the second viscosity coefficient, and $\mathbf{D}$ is the symmetric velocity gradient. Furthermore, $\Psi$ is the free energy, $\mu_\alpha$ are constituent chemical potentials, and $\mathbf{M}_{\mA\mB}$ and $m_{\mA\mB}$ are mobility parameters. Equation \eqref{eq: model intro: mom} describes the mixture momentum equation, equations \eqref{eq: model intro: mass} are the constituent mass balance laws, and \eqref{eq: model intro: chem} define the chemical potentials. 
%\newpage

A distinguishing feature of the model, in contrast to compressible two-phase flow models, is its preservation of thermodynamic consistency despite its reduced complexity (i.e. the combination of $N$ mass balance laws with just a single momentum balance law). Furthermore, the model connects to phase-field models and distinguishes itself from compressible two-phase flow models as follows. For specific closure models, it extends the Navier-Stokes-Korteweg (NSK) framework to multiple fluids, where the last two terms in \eqref{eq: model intro: mass} vanish in the single-fluid case. Additionally, under particular closure choices, these terms take the form of Cahn-Hilliard and Allen-Cahn-type contributions, establishing a direct link to NSCH(AC) models. Key distinctions from existing compressible two-phase flow models lie in (1) the absence of an explicit evolution equation for the volume fraction and (2) the formulation of the equation of state, as given in \eqref{eq: model intro: chem}.

\subsection{Plan of the paper}
The remainder of the paper is structured as follows. In \cref{sec: mix theory} we present the continuum theory of rational mechanics for compressible isothermal fluid mixtures. Next, in \cref{sec: 2nd law} we perform constitutive modeling via the Coleman-Noll procedure. Then, in \cref{sec: Properties} we discuss some properties of the model. In \cref{sec: hyperbolicity} we study the associated first-order system. Next, in \cref{sec: binary mixtures} we discuss the case of binary mixtures. Subsequently, in \cref{sec: Connections} we discuss connections to existing models. Finally, in \cref{sec: discussion} we conclude and provide further research directions.

\section{Continuum mixture theory}\label{sec: mix theory}

This section outlines the continuum theory for mixtures of compressible, isothermal constituents, serving as a foundational framework. While the theory is well known, its inclusion here is essential, particularly as it differs from the starting point for the design of compressible two-phase flow models. The outlined mixture theory description is closely aligned with \cite{ten2024thermodynamically,eikelder2024unified}, underscoring its suitability as a first-principles basis for the development for a wide range of models.

Continuum mixture theory is based on three general principles outlined in the groundbreaking work of Truesdell and Toupin \cite{truesdell1960classical,truesdell1984historical}:
\begin{enumerate}
    \item \textit{All properties of the mixture must be mathematical consequences of properties of the constituents.}
\item \textit{So as to describe the motion of a constituent, we may in imagination isolate it from the rest of the mixture, provided we allow properly for the actions of the other constituents upon it.}
\item \textit{The motion of the mixture is governed by the same equations as is a single body.}
\end{enumerate}
The first principle explains that the mixture is formed by its individual constituents. The second principle describes that the components of the physical model are interconnected through interaction terms. Lastly, the third principle states that the movement of the mixture cannot be distinguished from that of a single fluid.

In \cref{sec: prelim}, we provide the groundwork for the continuum theory of mixtures, covering essential kinematics. Following this, \cref{sec: BL} presents the evolution laws for individual constituents and mixtures.

\subsection{Preliminaries and kinematics}\label{sec: prelim}

The central concept of the continuum theory of mixtures posits that the material body consists of $N$ constituent bodies $\mathscr{B}_\mA$, where $\mA$ ranges from 1 to $N$. These constituent bodies $\mathscr{B}_\mA$ can simultaneously occupy the same region in space. Let $\mathbf{X}_\mA$ represent the spatial position of a particle of $\mathscr{B}_\mA$ in the Lagrangian (reference) configuration. The spatial position of a particle is described by the invertible deformation map:
\begin{align}
    \mathbf{x} := \bchi_{\mA}(\mathbf{X}_{\mA},t). 
\end{align}
The constituent partial mass density $\tilde{\rho}_{\mA}$ and specific mass density $\trho_\mA>0$ are respectively defined as:%\newpage
\begin{subequations}\label{eq: def trhoA rhoA}
  \begin{align}
  \tilde{\rho}_{\mA}(\bx,t) :=&~ \displaystyle\lim_{ \vert V \vert \rightarrow 0} \dfrac{M_{\mA}(V)}{\vert V \vert},\\
  \rho_\mA(\bx,t) :=&~ \displaystyle\lim_{\vert V_{\mA}\vert  \rightarrow 0} \dfrac{M_{\mA}(V)}{\vert V_{\mA}\vert },
\end{align}
\end{subequations}
where $V \subset \Omega$ (measure $\vert V \vert$)
is an arbitrary control volume around $\mathbf{x}$, $V_{\mA} \subset V$ (measure $\vert V_{\mA}\vert$) is the volume of constituent $\mA$ so that  $V =\cup_{\mA}V_{\mA}$. Furthermore, the constituents masses are $M_{\mA}=M_{\mA}(V)$, and the total mass in $V$ is $M=M(V)=\sum_{\mA}M_{\mA}(V)$. The mixture density is the superposition of the partial mass densities:
\begin{align}\label{eq: def rho}
\rho(\bx,t):=&~ \displaystyle\lim_{ \vert V \vert \rightarrow 0} \dfrac{M(V)}{\vert V \vert}=\displaystyle\sum_{\mA}\trho_\mA(\bx,t).
\end{align}
We introduce the mass fractions and volume fractions respectively as:
\begin{subequations}\label{eq: def c phi alpha}
  \begin{align}
    Y_\mA(\bx,t):=&~ \displaystyle\lim_{ \vert V \vert \rightarrow 0} \dfrac{M_{\mA}(V)}{M(V)}=\dfrac{\trho_\mA}{\rho},\\
    \phi_\mA(\bx,t):=&~ \displaystyle\lim_{ \vert V \vert \rightarrow 0} \dfrac{|V_{\mA}|}{|V|}=\dfrac{\trho_\mA}{\rho_\mA},
  \end{align}
\end{subequations}
which satisfy:
\begin{subequations}
  \begin{align}
    \displaystyle\sum_{\mA} Y_\mA(\bx,t)=&~ 1, \label{eq: sum c} \\
    \displaystyle\sum_{\mA} \phi_{\mA}(\bx,t)=&~ 1. \label{eq: sum phi}
  \end{align}
\end{subequations}

Next, we introduce the material time derivative $\grave{\uppsi}_\mA$ of the differentiable constituent function $\uppsi_\mA$:
\begin{align}
\grave{\uppsi}_\mA=\partial_t\uppsi_{\mA}(\mathbf{X}_{\mA},t) \vert_{\mathbf{X}_\mA},
\end{align}
where the notation $\vert_{\mathbf{X}_\mA}$ is used to emphasize that $\mathbf{X}_\mA$ is held fixed. The constituent velocity is constituent material derivative of the deformation map:
\begin{align}
\bv_{\mA}(\mathbf{x},t)=\partial_t\bchi_{\mA}(\mathbf{X}_{\mA},t) \vert_{\mathbf{X}_\mA} = \grave{\bchi}_\mA.
\end{align}
The mass-averaged mixture velocity $\bv$ (also called barycentric velocity) is defined via the following identification:
\begin{align}\label{eq: mix velo}
    \rho(\bx,t) \bv(\bx,t) = \displaystyle\sum_{\mA} \trho_\mA(\bx,t) \bv_\mA(\bx,t).
\end{align}
Furthermore, the (constituent) peculiar velocity is defined as:
\begin{align}\label{eq: def bwj}
    \bw_{\mA}(\bx,t):=\bv_{\mA}(\bx,t)-\bv(\bx,t),
\end{align}
and specifies the constituent velocity relative to the mixture's general motion. Introducing the (scaled) peculiar velocity:
\begin{align}\label{eq: def J and h}
    \bJ_\mA(\bx,t) :=&~ \trho_\mA(\bx,t) \bw_\mA(\bx,t),
\end{align}
we have the identity:
  \begin{align}\label{eq: rel gross motion zero}
    \displaystyle\sum_{\mA} \bJ_\mA(\bx,t)  =&~ 0.
  \end{align}
Finally, we introduce the material derivative of the mixture as: 
\begin{align}\label{eq: mat der}
    \dot{\uppsi}(\bx,t) =&~ \partial_t \uppsi(\bx,t) + \bv(\bx,t)\cdot \nabla \uppsi(\bx,t).
\end{align}

\subsection{Constituent balance laws}\label{sec: BL}

The motion of each constituent in the continuum theory of mixtures is expressed by an individualized set of balance laws, as outlined in the second general principle. These laws incorporate interaction terms that depict how the different constituents interact. The motion of each constituent $\mA = 1, \dots, N$ is governed by the following set of local balance laws for all $\mathbf{x}\in \Omega$ and $t >0$:
\begin{subequations}\label{eq: BL const}
  \begin{align}
        \partial_t \trho_\mA + {\rm div}(\trho_\mA \bv_\mA) &=~ \gamma_\mA, \label{eq: local mass balance constituent j} \\
        \partial_t (\trho_\mA\bv_\mA) + {\rm div} \left( \trho_\mA\bv_\mA\otimes \bv_\mA \right) -  {\rm div} \mathbf{T}_\mA -  \trho_\mA \mathbf{b}_\mA &=~ \boldsymbol{\pi}_\mA,\label{eq: lin mom constituent j}\\
        \mathbf{T}_\mA-\mathbf{T}_\mA^T &=~\mathbf{N}_\mA,\label{eq: ang mom constituent j}
  \end{align}
\end{subequations}
where $\Omega$ is the spatial domain.
Equation \eqref{eq: local mass balance constituent j} describes the local constituent mass balance law, \eqref{eq: lin mom constituent j} the local constituent linear momentum balance law, and \eqref{eq: ang mom constituent j} the local constituent angular momentum balance. The interaction terms are $\gamma_{\mA}, \bpi_\mA$ and $\mathbf{N}_\mA$. These denote respectively the mass supply of constituent $\mA$ due to chemical reactions with the other constituents, the momentum exchange rate of constituent $\mA$, and the intrinsic moment of momentum of constituent $\mA$. Furthermore, $\mathbf{T}_\mA$ is the Cauchy stress tensor of constituent $\mA$, $\mathbf{b}_\mA$ the constituent external body force. In this article we assume equal body forces ($\mathbf{b}_\mA= \mathbf{b}$ for $\mA = 1, \dots, N$) and restrict to body forces of gravitational type: $\mathbf{b} = -b \boldsymbol{\jmath} = -b \nabla y$, with $y$ the vertical coordinate, $\boldsymbol{\jmath}$ the vertical unit vector and $b$ a constant.

In addition, one can infer that the constituent mass balance law can be written as:
\begin{align}\label{eq: BL mass constituent material der}
    \partial_t \trho_\mA + {\rm div}(\trho_\mA \bv) + {\rm div} \mathbf{J}_\mA &=~ \gamma_\mA.
\end{align}
We denote the kinetic and gravitational energies of the constituents respectively as:
\begin{subequations}
    \begin{align}
  \mathscr{K}_\mA =&~\trho_\mA \|\bv_\mA\|^2/2,\\
  \mathscr{G}_\mA =&~\trho_\mA b y,
\end{align}
\end{subequations}
where $\|\mathbf{v}_\mA\|=(\mathbf{v}_\mA \cdot \mathbf{v}_\mA)^{1/2}$ is the Euclidean norm of the velocity $\mathbf{v}_\mA$.

\begin{remark}[Volume-averaged mixture velocity]
  Alternative to the mass-averaged mixture velocity \eqref{eq: mix velo}, another commonly adopted mixture velocity is the volume-averaged velocity $\bu$ defined as $\bu = \sum_\mA \phi_\mA\bv_\mA$. In the case of incompressible flow ($\rho_\mA(\mathbf{x},t)= \rho_\mA$) without mass transfer ($\gamma_\mA=0$) one may deduce from \eqref{eq: sum phi} and \eqref{eq: local mass balance constituent j} that it is divergence free: ${\rm div} \bu = 0$.
\end{remark}

\subsection{Mixture balance laws}
Next, we proceed to the continuum balance laws of the mixtures. In agreement with the first general principle, the motion of the mixture results from that of the individual constituents. Summing the individual balance laws \eqref{eq: BL const} across all constituents yields:%\newpage
\begin{subequations}\label{eq: BL mix}
  \begin{align}
        \partial_t \rho + {\rm div}(\rho \bv) &=~ 0, \label{eq: local mass balance mix} \\
        \partial_t (\rho \bv) + {\rm div} \left( \rho \bv \otimes \bv \right) -  {\rm div} \mathbf{T} -  \rho \mathbf{b} &=~0,\label{eq: lin mom mix0}\\
        \mathbf{T}-\mathbf{T}^T &=~0,\label{eq: ang mom mix0}
  \end{align}
\end{subequations}
where
\begin{subequations}
    \begin{align}
       \mathbf{T} :=&~ \sum_\mA \mathbf{T}_\mA-\trho_\mA\bw_\mA\otimes\bw_\mA,\\
    \mathbf{b} :=&~\frac{1}{\rho}\sum_\mA \trho_\mA\mathbf{b}_\mA.
    \end{align}
\end{subequations}
In agreement with the third general principle, we have postulated the following balance conditions:
\begin{subequations}
  \begin{align}
      \displaystyle\sum_\mA \gamma_\mA  =&~ 0,\label{eq: balance mass fluxes}\\
      \displaystyle\sum_\mA \boldsymbol{\pi}_\mA =&~ 0,\label{eq: balance momentum fluxes}\\
      \displaystyle\sum_\mA \mathbf{N}_\mA =&~ 0.
      \end{align}
\end{subequations}
The first general principle of mixture theory additionally implies that the kinetic and gravitational energy of the mixture are the superposition of the constituent energies:
\begin{subequations}
  \begin{align}
  \mathscr{K} =&~ \displaystyle\sum_{\mA} \mathscr{K}_\mA,\label{eq: def sum K}\\
  \mathscr{G} =&~ \displaystyle\sum_{\mA} \mathscr{G}_\mA.
\end{align}
\end{subequations}
\begin{remark}[Kinetic energy]\label[remark]{rmk: split kin energy}
    The kinetic energy of the mixture can be decomposed as:
\begin{subequations}\label{eq: relation kin energies}
    \begin{align}
      \mathscr{K} =&~ \bar{\mathscr{K}} + \displaystyle\sum_{\mA} \frac{1}{2} \trho_\mA \|\mathbf{w}_\mA\|^2,\\
       \bar{\mathscr{K}} =&~ \frac{1}{2} \rho \|\mathbf{v}\|^2,\label{eq: kin avg}
\end{align}
\end{subequations}
where $\bar{\mathscr{K}}$ is a kinetic energy quantity expressed in pure mixture variables, and where the second member is comprised of kinetic energies based on peculiar velocities. 
\end{remark}

\section{Constitutive modeling}\label{sec: 2nd law}

This section is devoted to the construction of constitutive models, employing a structured approach akin to  \cite{ten2024thermodynamically,eikelder2024unified}. First, \cref{sec: const mod: subsec: def} provides assumptions, and details on the Coleman-Noll modeling procedure \citep{coleman1974thermodynamics}. Following this, in \cref{sec: const mod: subsec: model restr}, we determine the constitutive modeling restriction following from \cref{sec: const mod: subsec: def}. Finally, in \cref{sec: const mod: subsec: select}, we identify specific constitutive models that align with these restrictions.

\subsection{Assumptions and modeling choices}\label{sec: const mod: subsec: def}

Instead of working with the full set of balance laws \eqref{eq: local mass balance constituent j}--\eqref{eq: ang mom constituent j}, we restrict to the reduced set:
\begin{subequations}\label{eq: BL constitutive}
  \begin{align}
        \partial_t \trho_\mA + {\rm div}(\trho_\mA \bv) + {\rm div} \bH_\mA &=~ \zeta_\mA, \label{eq: BL constitutive: mass} \\
        \partial_t (\rho \bv) + {\rm div} \left( \rho\bv\otimes \bv \right) -  {\rm div} \mathbf{T} -  \rho \mathbf{b} &=~0,\label{eq: lin mom mix}\\
        \mathbf{T}-\mathbf{T}^T &=~0,\label{eq: ang mom mix}
  \end{align}
\end{subequations}
and where $\mA=1,...,N$ in \eqref{eq: BL constitutive: mass}. We have decomposed the mass transfer terms into conservative and possible non-conservative parts via $\gamma_\mA = \zeta_\mA - {\rm div} \mathbf{j}_\mA$ with $\bH_\mA := \bJ_\mA + \bj_\mA$. The state variables in the system \eqref{eq: BL constitutive} are $\trho_\mA$ ($\mA = 1,...,N$) and $\bv$. We seek constitutive models for $\mathbf{T}$, $\bH_\mA$ and $\zeta_\mA$ ($\mA = 1,...,N$). Hereby we discard the definition \eqref{eq: def J and h}, but enforce the balance condition \eqref{eq: rel gross motion zero}.

\begin{remark}[Decomposition mass transfer]
  Rather than modeling $\bH_\mA$ and $\zeta_\mA$, one may work directly with the original terms $\bJ$ and $\gamma_\mA$, as both approaches yield equivalent closure models. The decomposition $\gamma_\mA = \zeta_\mA - \operatorname{div}\mathbf{j}_\mA$ is introduced solely for interpretative purposes. For example, assuming identical constituent velocities, $\bv_\mA = \bv$, does not conflict with the constitutive model for $\bH_\mA$.
\end{remark}
%\begin{remark}[Classification]
%    The above assumptions signify that the resulting model is comprised of $N$ constituent mass balance laws, and a single momentum balance law. This combination is typically adopted in NSCHAC models, and the closest match with the classification of \cite{hutter2013continuum} is a class-I model.
%\end{remark}

We postulate the energy-dissipation law:
\begin{align}\label{eq: energy dissipation}
    \dfrac{{\rm d}}{{\rm d}t} \mathscr{E} = \mathscr{W} - \mathscr{D},
\end{align}
as a design principle for constructing closure models for $\mathbf{T}, \bH_\mA$ and $\zeta_\mA$ ($\mA=1,...,N$). The total energy is composed of the Helmholtz free energy, the kinetic energy and the gravitational energy:
\begin{align}\label{eq: total energy}
  \mathscr{E} =  \displaystyle\int_{\mathcal{R}(t)}(\Psi + \bar{\mathscr{K}} + \mathscr{G})~{\rm d}v.
\end{align}
In this formulation, $\mathcal{R}(t) \subset \Omega$ is an arbitrary, time-dependent control volume with volume element ${\rm d}v$ and unit outward normal $\boldsymbol{\nu}$, transported by the velocity field $\bv$. Moreover, $\mathscr{W}$ denotes the work rate applied on the boundary $\partial\mathcal{R}(t)$, while $\mathscr{D}$ represents the internal dissipation, which is required to satisfy $\mathscr{D}\geq 0$.

\begin{remark}[Energy-Dissipation Statement]
In the energy-dissipation law \eqref{eq: energy dissipation}, the kinetic energy term is taken as $\bar{\mathscr{K}}$, rather than the full kinetic energy of the mixture $\mathscr{K}$ (see \cref{rmk: split kin energy}). Consequently, \eqref{eq: energy dissipation} can be seen as a simplified formulation of the second law of thermodynamics for mixtures (see, e.g., \cite{eikelder2023unified,ten2024thermodynamically}). Moreover, the use of control volumes is not essential; an equivalent modeling constraint, as discussed in \cref{sec: const mod: subsec: model restr}, may be derived through analogous steps at the local PDE level.
\end{remark}

In this article we postulate the free energy to belong to the constitutive class:
\begin{align}\label{eq: class Psi}
  \Psi = \hat{\Psi}\left(\left\{\trho_\mA\right\}_{\mA=1,...,N},\left\{\nabla \trho_\mA\right\}_{\mA=1,...,N}\right), 
\end{align}
where we note that both $\left\{\trho_\mA\right\}_{\mA=1,...,N}$ and $\left\{\nabla \trho_\mA\right\}_{\mA=1,...,N}$ consist of independent variables. In addition, we introduce the constituent chemical potentials:
\begin{align}
    \hat{\mu}_\mA =&~ \dfrac{ \partial \hat{\Psi}}{\partial \trho_\mA} - {\rm div}\dfrac{\partial \hat{\Psi}}{\partial\nabla \trho_\mA}.
\end{align}%\newpage
\begin{remark}[Relation to compressible two-phase flow models]
    The modeling choices presented in this section distinguish our approach from conventional compressible two-phase flow models in several important ways. First, our formulation is based on the reduced set of balance laws \eqref{eq: BL constitutive}. In contrast, compressible two-phase flow models typically include additional evolution equations for the volume fractions, which do not naturally arise from the continuum mixture theory described in \cref{sec: mix theory}. Consequently, while the state variables in compressible two-phase flow models consist of specific densities $\{\rho_\mA\}$ and volume fractions $\{\phi_\mA\}$ (subject to the saturation constraint \eqref{eq: sum phi}), our approach employs partial densities $\{\trho_\mA\}$ instead. Finally, the constitutive class \eqref{eq: class Psi} depends on all partial densities, thereby enabling the modeling of interaction (Van der Waals) forces—an aspect that is absent in standard compressible two-phase flow models. %We return to this latter point in \cref{rmk: absence interaction forces}.
\end{remark}

\begin{remark}[Constitutive class free energy]
In many conventional phase-field models, including those of the Cahn-Hilliard type, the free energy is expressed in terms of order parameters such as volume fractions or mass fractions (often referred to as concentrations). In contrast, the constitutive class \eqref{eq: class Psi} is formulated in terms of partial densities. We explore the relationship between these approaches in the binary case in \cref{sec: binary mixtures} and provide the corresponding derivation in \cref{sec: appendix: alternative constitutive modeling: rho c Psi}.
\end{remark}

\subsection{Modeling restriction}\label{sec: const mod: subsec: model restr}
We proceed with the evaluation of the evolution of the energy \eqref{eq: total energy}. By applying Reynolds transport theorem to the free energy $\hat{\Psi}$ we have:
\begin{align}
      \dfrac{{\rm d}}{{\rm d}t}\displaystyle\int_{\mathcal{R}(t)} \hat{\Psi} ~{\rm d}v = \displaystyle\int_{\mathcal{R}(t)} \partial_t \hat{\Psi} ~{\rm d}v + \displaystyle\int_{\partial \mathcal{R}(t)} \hat{\Psi} \bv \cdot \bnu  ~{\rm d}a.
\end{align}
We apply the divergence theorem and expand the derivatives:
\begin{align}\label{eq: Psi derivation 1}
    \dfrac{{\rm d}}{{\rm d}t}\displaystyle\int_{\mathcal{R}(t)} \hat{\Psi} ~{\rm d}v  = \displaystyle\int_{\mathcal{R}(t)} &~ \hat{\Psi} ~{\rm div} \bv + \sum_\mA\dfrac{\partial \hat{\Psi}}{\partial \trho_\mA} \dot{\trho}_\mA+ \dfrac{\partial \hat{\Psi}}{\partial \nabla \trho_\mA}\cdot \left(\nabla \trho_\mA\right)\dot{} ~{\rm d}v.
\end{align}
Substituting the following identity for the material derivative:
\begin{align}\label{eq: relation grad phi}
    (\nabla \omega)\dot{} = \nabla (\dot{\omega}) - (\nabla \omega)^T\nabla \bv,
\end{align}
with $\omega = \trho_\mA$ into \eqref{eq: Psi derivation 1} yields:
\begin{align}\label{eq: Psi derivation 2}
    \dfrac{{\rm d}}{{\rm d}t}\displaystyle\int_{\mathcal{R}(t)} \hat{\Psi} ~{\rm d}v  = \displaystyle\int_{\mathcal{R}(t)} &~ \hat{\Psi} ~{\rm div} \bv + \sum_\mA\dfrac{\partial \hat{\Psi}}{\partial \trho_\mA} \dot{\trho}_\mA+ \dfrac{\partial \hat{\Psi}}{\partial \nabla \trho_\mA}\cdot \nabla (\dot{\trho}_\mA) \nn\\
    &~- \dfrac{\partial \hat{\Psi}}{\partial \nabla \trho_\mA}\cdot (\nabla \trho_\mA)^T\nabla \bv ~{\rm d}v.
\end{align}
By subsequently integrating by parts, we arrive at:
\begin{align}
    \dfrac{{\rm d}}{{\rm d}t}\displaystyle\int_{\mathcal{R}(t)} \hat{\Psi} ~{\rm d}v  =&~ \displaystyle\int_{\mathcal{R}(t)} \hat{\Psi}~{\rm div} \bv + \sum_\mA \hat{\mu}_\mA \dot{\trho}_\mA  - \nabla \trho_\mA \otimes \dfrac{\partial \hat{\Psi}}{\partial \nabla \trho_\mA}: \nabla \bv  ~{\rm d}v  \nn\\
    &~+ \displaystyle\int_{\partial \mathcal{R}(t)}\sum_\mA \dot{\trho}_\mA \dfrac{\partial \hat{\Psi}}{\partial \nabla \trho_\mA}\cdot \boldsymbol{\nu} ~{\rm d}a.
\end{align}
Substituting the constituent mass balance laws \eqref{eq: BL mass constituent material der} and applying integration by parts leads to:
\begin{align}\label{eq: Psi 2}
    \dfrac{{\rm d}}{{\rm d}t}\displaystyle\int_{\mathcal{R}(t)} \hat{\Psi} ~{\rm d}v = &~ \displaystyle\int_{\mathcal{R}(t)} \hat{\Psi}~{\rm div} \bv -\sum_\mA \hat{\mu}_\mA\trho_\mA {\rm div}\bv +\sum_\mA \nabla \hat{\mu}_\mA \cdot \bH_\mA \nn\\
    &~~~- \sum_\mA \nabla \trho_\mA \otimes \dfrac{\partial \hat{\Psi}}{\partial \nabla \trho_\mA}: \nabla \bv +\sum_\mA    \hat{\mu}_\mA \zeta_\mA~{\rm d}v\nn\\
    &~+ \displaystyle\int_{\partial \mathcal{R}(t)}\left(\sum_\mA \dot{\trho}_\mA \dfrac{\partial \hat{\Psi}}{\partial \nabla \trho_\mA}-\hat{\mu}_\mA \bH_\mA\right)\cdot \boldsymbol{\nu} ~{\rm d}a.
\end{align}
Subsequently, the evolution equations for the kinetic and gravitational energies follow straightforward from \eqref{eq: BL constitutive: mass} and \eqref{eq: lin mom mix} (see \cite{eikelder2023unified} for details):
\begin{subequations}\label{eq: kin grav evo}
    \begin{align}
    \dfrac{{\rm d}}{{\rm d}t}\displaystyle\int_{\mathcal{R}(t)} \mathscr{K} ~{\rm d}v =&~ \displaystyle\int_{\mathcal{R}(t)} - \nabla \bv : \mathbf{T}+\rho\bv\cdot\mathbf{g}    ~{\rm d}v+ \displaystyle\int_{\partial \mathcal{R}(t)} \bv \cdot \mathbf{T} \bnu  ~{\rm d}a,\\
    \dfrac{{\rm d}}{{\rm d}t}\displaystyle\int_{\mathcal{R}(t)} \mathscr{G} ~{\rm d}v =&~- \displaystyle\int_{\mathcal{R}(t)} \rho \bv\cdot\mathbf{g}~{\rm d}v.
  \end{align}
  \end{subequations}
Summing \eqref{eq: Psi 2} and \eqref{eq: kin grav evo} yields:
\begin{align}\label{eq: second law subst 1}
    \dfrac{{\rm d}}{{\rm d}t} \mathscr{E} = &~ \displaystyle\int_{\partial \mathcal{R}(t)}\left(\bv^T\mathbf{T}-\sum_\mA \left(\hat{\mu}_\mA\bH_\mA -\dot{\trho}_\mA \dfrac{\partial \hat{\Psi}}{\partial \nabla \trho_\mA}\right)\right)\cdot \boldsymbol{\nu} ~{\rm d}a \nn\\
    &~- \displaystyle\int_{\mathcal{R}(t)}   \left(\mathbf{T}  +\sum_\mA \nabla \trho_\mA \otimes \dfrac{\partial \hat{\Psi}}{\partial \nabla \trho_\mA} +\left(\sum_\mA \hat{\mu}_\mA\trho_\mA -  \hat{\Psi}\right)\mathbf{I}\right):\nabla \mathbf{v}\nn\\
    &~~~~~~~~~~~+ \sum_\mA\left(-\nabla \hat{\mu}_\mA\cdot \bH_\mA  - \hat{\mu}_\mA\zeta_\mA \right)~{\rm d}v,
\end{align}
where we identify the rate of work and the dissipation as:
\begin{subequations}\label{eq: W, D}
\begin{align}
    \mathscr{W} =&~ \displaystyle\int_{\partial \mathcal{R}(t)}\left(\bv^T\mathbf{T}-\sum_\mA \left(\hat{\mu}_\mA\bH_\mA -\dot{\trho}_\mA \dfrac{\partial \hat{\Psi}}{\partial \nabla \trho_\mA}\right)\right)\cdot \boldsymbol{\nu} ~{\rm d}a,\\
    \mathscr{D} =&~ \displaystyle\int_{\mathcal{R}(t)}   \left(\mathbf{T}  +\sum_\mA \nabla \trho_\mA \otimes \dfrac{\partial \hat{\Psi}}{\partial \nabla \trho_\mA} +\left(\sum_\mA \hat{\mu}_\mA\trho_\mA -  \hat{\Psi}\right)\mathbf{I}\right):\nabla \mathbf{v}\nn\\
    &~~~~~~~~~~~+ \sum_\mA\left(-\nabla \hat{\mu}_\mA\cdot \bH_\mA  - \hat{\mu}_\mA\zeta_\mA \right)~{\rm d}v.\label{eq: def diffusion}
\end{align}
\end{subequations}
Since the control volume $\mathcal{R}=\mathcal{R}(t)$ is arbitrary, the energy-dissipation law is satisfied provided that the following local inequality holds:
\begin{align}\label{eq: second law 4}
     \left(\mathbf{T}  +\sum_\mA \nabla \trho_\mA \otimes \dfrac{\partial \hat{\Psi}}{\partial \nabla \trho_\mA} +\left(\sum_\mA \hat{\mu}_\mA\trho_\mA -  \hat{\Psi}\right)\mathbf{I}\right):\nabla \mathbf{v}&\nn\\
    - \sum_\mA \bH_\mA\cdot \nabla \hat{\mu}_\mA  - \sum_\mA \zeta_\mA\hat{\mu}_\mA &\geq 0.
\end{align}

%\newpage
\subsection{Selection of constitutive models}\label{sec: const mod: subsec: select}
By means of the Colemann-Noll procedure, the inequality \eqref{eq: second law 4} permits to restrict to mixture stress tensors $\mathbf{T}$, constituent diffusive fluxes $\bH_\mA$ and constituent mass fluxes $\zeta_\mA$ that belong to the constitutive classes:
\begin{subequations}\label{eq: class T J gamma}
\begin{align}
    \mathbf{T} =&~ \hat{\mathbf{T}}\left(\nabla \bv, \left\{\trho_\mB\right\}, \left\{\nabla \trho_\mB\right\}, \left\{\hat{\mu}_\mB\right\},\left\{\nabla \hat{\mu}_\mB\right\}\right),\label{eq: class T}\\
    \bH_\mA =&~ \hat{\bH}_\mA\left(\left\{\trho_\mB\right\}, \left\{\nabla \trho_\mB\right\}, \left\{\hat{\mu}_\mB\right\},\left\{\nabla \hat{\mu}_\mB\right\}\right),\label{eq: class J}\\
    \zeta_\mA =&~ \hat{\zeta}_\mA\left(\left\{\trho_\mB\right\}, \left\{\hat{\mu}_\mB\right\}\right).\label{eq: class gamma}
\end{align}
\end{subequations}
%\newpage
We seek for constitutive models \eqref{eq: class T J gamma} that render each of the three terms in \eqref{eq: second law 4} non-negative:
\begin{subequations}\label{eq: restrictions}
    \begin{align}
     \left(\hat{\mathbf{T}}  +\sum_\mA \nabla \trho_\mA \otimes \dfrac{\partial \hat{\Psi}}{\partial \nabla \trho_\mA} +\left(\sum_\mA \hat{\mu}_\mA\trho_\mA -  \hat{\Psi}\right)\mathbf{I}\right):\nabla \mathbf{v} &\geq 0, \label{eq: restrictions T}\\
     -\sum_\mA \hat{\bH}_\mA\cdot \nabla \hat{\mu}_\mA &\geq 0, \label{eq: restrictions J}\\    
     -\sum_\mA \hat{\zeta}_\mA\hat{\mu}_\mA &\geq 0. \label{eq: restrictions gamma}
\end{align}
\end{subequations}

\noindent \textit{Constitutive choices}. We make the following constitutive choices:
\begin{subequations}\label{eq:constitutive_choices}
\begin{align}
\hat{\mathbf{T}} &= -\sum_\mA \nabla \trho_\mA \otimes \dfrac{\partial \hat{\Psi}}{\partial \nabla \trho_\mA} - \left(\sum_\mA \hat{\mu}_\mA\trho_\mA - \hat{\Psi}\right)\mathbf{I} + \nu \left(2\mathbf{D}+\lambda({\rm div}\,\bv)\mathbf{I}\right), \label{eq:stress_tensor_choice} \\
\hat{\bH}_\mA &= -\sum_\mB \mathbf{M}_{\mA\mB} \nabla \hat{\mu}_\mB, \label{eq:model_J} \\
\hat{\zeta}_\mA &= -\sum_\mB m_{\mA\mB}\hat{\mu}_\mB, \label{eq:const_model_mass_flux}
\end{align}
\end{subequations}
where $\nu\geq 0$ is the dynamic viscosity, $\lambda\geq -2/d$ (with $d$ denoting the number of dimensions), $\mathbf{M}_{\mA\mB}$ is a symmetric positive definite tensor, with the same dependencies as \eqref{eq: class J}, satisfying 
\begin{align}
\sum_\mA \mathbf{M}_{\mA\mB} = \sum_\mB \mathbf{M}_{\mA\mB} = 0,\quad \mathbf{M}_{\mA\mB}|_{Y_\mC=1}=0\quad\text{for all } \mC=1,\dots,N,
\end{align}
and $m_{\mA\mB}$ is a symmetric positive definite scalar mobility, with the same dependencies as \eqref{eq: class gamma}, satisfying
\begin{align}
\sum_\mA m_{\mA\mB} = \sum_\mB m_{\mA\mB} = 0,\quad m_{\mA\mB}|_{Y_\mC=1}=0\quad\text{for all } \mC=1,\dots,N.
\end{align}
Moreover, compatibility with the angular momentum condition \eqref{eq: ang mom mix} requires
\begin{align}
\nabla \trho_\mA \otimes \dfrac{\partial \hat{\Psi}}{\partial \nabla \trho_\mA} 
= \dfrac{\partial \hat{\Psi}}{\partial \nabla \trho_\mA} \otimes \nabla \trho_\mA.
\end{align}

\begin{lemma}[Compatibility of constitutive choices]\label[lemma]{lem:compatibility_constitutive_choices}
The choices \eqref{eq:stress_tensor_choice}--\eqref{eq:const_model_mass_flux} are compatible with the restrictions \eqref{eq: restrictions T}--\eqref{eq: restrictions gamma}, as well as with the balance of relative gross motion \eqref{eq: rel gross motion zero} and the balance of mass supply \eqref{eq: balance mass fluxes}.
\end{lemma}
\begin{proof}
For the stress tensor, a direct calculation yields
\begin{align*}
\Biggl(\hat{\mathbf{T}} &+\sum_\mA \nabla \trho_\mA \otimes \dfrac{\partial \hat{\Psi}}{\partial \nabla \trho_\mA} + \Bigl(\sum_\mA \hat{\mu}_\mA\trho_\mA - \hat{\Psi}\Bigr)\mathbf{I}\Biggr):\nabla \mathbf{v} \\
&= 2 \nu \Biggl(\mathbf{D} - \frac{1}{d} ({\rm div}\,\mathbf{v})\mathbf{I}\Biggr):\Biggl(\mathbf{D} - \frac{1}{d} ({\rm div}\,\mathbf{v})\mathbf{I}\Biggr)
+ \nu\Biggl(\lambda + \frac{2}{d}\Biggr) ({\rm div}\,\mathbf{v})^2 \ge 0,
\end{align*}
which confirms compatibility with \eqref{eq: restrictions T}. For the diffusive flux, the condition $\sum_\mA \mathbf{M}_{\mA\mB}=0$ guarantees compatibility with the balance \eqref{eq: rel gross motion zero}, while the positive semi-definiteness of $\mathbf{M}_{\mA\mB}$ ensures the thermodynamical restriction \eqref{eq: restrictions J}. Similarly, for the mass transfer term, the properties $\sum_\mA m_{\mA\mB}=0$ and the positive definiteness of $m_{\mA\mB}$ ensure compatibility with the mass supply balance \eqref{eq: balance mass fluxes} and the restriction \eqref{eq: restrictions gamma}.
\end{proof}

\begin{remark}[Mass transfer in compressible two-phase flow models]\label{rmk: special case mobility}  
We consider the special case:  
\begin{align}  
    m_{\mA\mB} = \begin{cases}  
        -\hat{m}_{\mA\mB}, & \text{if } \mA \neq \mB, \\  
        \sum_{\mC\neq \mA} \hat{m}_{\mA\mC}, & \text{if } \mA = \mB.  
    \end{cases}  
\end{align}  
Under this assumption, we obtain:  
\begin{align}  
    \hat{\zeta}_\mA &= - \sum_{\mB\neq \mA} m_{\mA\mB} \nabla \hat{\mu}_\mB - m_{\mA\mA} \nabla \hat{\mu}_\mA \nn\\  
    &= \sum_{\mB\neq \mA} \hat{m}_{\mA\mB} \nabla \hat{\mu}_\mB - \sum_{\mC\neq\mA} \hat{m}_{\mA\mC} \nabla \hat{\mu}_\mA \nn\\  
    &= -\sum_{\mB} \hat{m}_{\mA\mB} \nabla (\hat{\mu}_\mA - \hat{\mu}_\mB).  
\end{align}  
This formulation is commonly used for binary flows in compressible two-phase models; see, e.g. \cite{coquel2013modelling,saurel2018diffuse,pelanti2022arbitrary}.  
\end{remark}  

This concludes the fundamental exploration of constitutive models compatible with the imposed energy-dissipative modeling restriction. Substitution provides the following model: 
\begin{subequations}\label{eq: model full}
  \begin{align}
   \partial_t (\rho \bv) + {\rm div} \left( \rho \bv\otimes \bv \right) + \nabla p + {\rm div} \left(\sum_\mB \nabla \trho_\mB \otimes \dfrac{\partial \hat{\Psi}}{\partial \nabla \trho_\mB} \right)& \nn\\
    - {\rm div} \left(   \nu (2\mathbf{D}+\lambda({\rm div}\bv) \mathbf{I}) \right)-\rho\mathbf{g} &=~ 0, \label{eq: model full: mom}\\
  \partial_t \trho_\mA  + {\rm div}(\trho_\mA  \bv) - \sum_{\mB} {\rm div} \left(\mathbf{M}_{\mA\mB}\nabla\hat{\mu}_\mB\right)  + \sum_{\mB} m_{\mA\mB} \hat{\mu}_\mB &=~0,\label{eq: model full: mass}\\
  \hat{\mu}_\mA - \dfrac{\partial \Psi}{\partial \trho_\mA} + {\rm div} \left(  \dfrac{\partial \Psi}{\partial \nabla \trho_\mA} \right)&=~0,
  \end{align}
\end{subequations}
where the pressure $p$ takes the form:
\begin{align}
    p=\sum_\mB (\hat{\mu}_\mB\trho_\mB) -\hat{\Psi}.
\end{align}

%\newpage
\begin{remark}[Relation to Cahn-Hilliard/Allen-Cahn/Korteweg models]
    The framework connects to Cahn-Hilliard/Allen-Cahn/Korteweg-type models by selecting\\ $\hat{\Psi} = \hat{\Psi}_0+\tfrac{1}{2}\sum_{\mB\mC} \lambda_{\mB\mC} \nabla \trho_\mB\cdot\nabla \trho_\mC$, where $\hat{\Psi}_0=\hat{\Psi}_0(\left\{\trho_\mB\right\})$ and where $\lambda_{\mB\mC}$ are model parameters.
\end{remark}
\begin{remark}[Alternative pressure variable]
  A common alternative approach is to decompose the stress tensor into its volumetric (isotropic) and deviatoric (traceless) components, which allows one to define the volumetric pressure $p^{\rm vol}$:
  \begin{subequations}
      \begin{align}
  \hat{\mathbf{T}} =&~ -p^{\rm vol} \mathbf{I} + \hat{\mathbf{T}}^{\rm dev}, \\
  p^{\rm vol} =&~ \sum_\mB \left(\hat{\mu}_\mB\trho_\mB  - \frac{1}{d} \nabla \trho_\mB \cdot \dfrac{\partial \hat{\Psi}}{\partial \nabla \trho_\mB}\right) -\hat{\Psi} - \nu \left(\lambda + \frac{2}{d}\right){\rm div} \bv,\\
  \hat{\mathbf{T}}^{\rm dev} = &~\sum_\mB\left( \nabla \trho_\mB \otimes \dfrac{\partial \hat{\Psi}}{\partial \nabla \trho_\mB}- \frac{1}{d}\nabla \trho_\mB \cdot \dfrac{\partial \hat{\Psi}}{\partial \nabla \trho_\mB} \mathbf{I}\right) + 2 \nu \left( \mathbf{D} - \frac{1}{d} ({\rm div} \mathbf{v}) \mathbf{I}\right).
\end{align}
\end{subequations}
\end{remark}

We explicitly state the compatibility with the energy dissipation condition.
\begin{theorem}[Compatibility energy dissipation]\label{eq: compatibility energy dissipation}
The model \eqref{eq: model full} is compatible with the energy dissipation condition \eqref{eq: energy dissipation}.
\end{theorem}
\begin{proof}
  This follows from \cref{lem:compatibility_constitutive_choices}, in particular we have:
\begin{align}\label{eq: second law red 7}
    \mathscr{D} = &~ \displaystyle\int_{\mathcal{R}(t)} 2 \nu \left( \mathbf{D} - \frac{1}{d} ({\rm div} \mathbf{v}) \mathbf{I}\right):\left(\mathbf{D} - \frac{1}{d} ({\rm div} \mathbf{v}) \mathbf{I}\right)+ \nu \left(\lambda + \frac{2}{d}\right)\left({\rm div} \mathbf{v}\right)^2\nn\\
    &~\quad\quad + \displaystyle\sum_{\mA,\mB} M_{\mA\mB}\nabla \hat{\mu}_\mA\cdot\nabla \hat{\mu}_\mB +\sum_{\mA,\mB} m_{\mA\mB}\hat{\mu}_\mA\hat{\mu}_\mB~{\rm d}v \geq 0.
\end{align}
\end{proof}

\section{Properties}\label{sec: Properties}
We first discuss compact formulations in \cref{subsec: compact formulations}. Next, we discuss the decomposition of the free energy in \cref{subsec: decomp free energy}.
Finally, in \cref{subsec: equilibrium conditions} we present the equilibrium conditions of the model. %Finally, in \cref{subsec: hyperbolicity} we provide details on the hyperbolic nature of the associated first-order system.
\subsection{Compact formulations}\label{subsec: compact formulations}
Simplification of the model \eqref{eq: model full} can be achieved by employing the following identity.
\begin{lemma}[Identity Korteweg stresses]\label[lemma]{lem: identity Korteweg stresses}
We have the following identity for the pressure and Korteweg stresses:
  \begin{align}
      \nabla p + {\rm div} \left( \sum_\mA\nabla \trho_\mA \otimes \dfrac{\partial \hat{\Psi}}{\partial \nabla \trho_\mA}\right) = \sum_\mA\trho_\mA\nabla\hat{\mu}_\mA.
  \end{align}
\end{lemma}
\begin{proof}
See \cref{lem: korteweg id}.
\end{proof}
\begin{remark}
  The identity in \cref{lem: identity Korteweg stresses} represents the multi-component extension of the well-known identity for the Korteweg stress.
\end{remark}
Applying \cref{lem: identity Korteweg stresses} we find the more compact form:
\begin{subequations}\label{eq: model simplified}
  \begin{align}
   \partial_t (\rho \bv) + {\rm div} \left( \rho \bv\otimes \bv \right) + \sum_{\mB} \trho_\mB\nabla\hat{\mu}_\mB %& \nn\\
    - {\rm div} \left(   \nu (2\mathbf{D}+\lambda({\rm div}\bv) \mathbf{I}) \right)-\rho\mathbf{g} &=~ 0, \label{eq: eq: model  simplified: mom}\\
  \partial_t \trho_\mA  + {\rm div}(\trho_\mA  \bv) - \sum_{\mB} {\rm div} \left(M_{\mA\mB}\nabla\hat{\mu}_\mB\right)  + \sum_{\mB} m_{\mA\mB} \hat{\mu}_\mB &=~0.\label{eq: eq: model simplified: mass}%\\
  %\hat{\mu}_\mA - \dfrac{\partial \Psi_\mA}{\partial \trho_\mA} + {\rm div} \left(  \dfrac{\partial \Psi_\mA}{\partial \nabla \trho_\mA} \right)&=~0,
  \end{align}
\end{subequations}
Furthermore, we note that the model may alternatively be written in a form that more closely links to existing phase-field models:
\begin{subequations}\label{eq: model alt}
  \begin{align}
   \partial_t (\rho \bv) + {\rm div} \left( \rho \bv\otimes \bv \right) + \rho \sum_{\mB}  Y_\mB\nabla\hat{\mu}_\mB 
    - {\rm div} \left(   \nu (2\mathbf{D}+\lambda({\rm div}\bv) \mathbf{I}) \right)-\rho\mathbf{g} &=~ 0, \label{eq: model alt: mom}\\
  \partial_t \rho + {\rm div}(\rho \bv) &=~0,\label{eq: eq: model alt: mass full}\\
  \partial_t (\rho Y_\mA)  + {\rm div}(\rho Y_\mA  \bv) - \sum_{\mB} {\rm div} \left(M_{\mA\mB}\nabla\hat{\mu}_\mB\right)  + \sum_{\mB} m_{\mA\mB} \hat{\mu}_\mB &=~0,\label{eq: eq: model alt: mass}%\\
  %\hat{\mu}_\mA - \dfrac{\partial \Psi_\mA}{\partial (\rho Y_\mA)} + {\rm div} \left(  \dfrac{\partial \Psi_\mA}{\partial \nabla (\rho Y_\mA)} \right)&=~0,
  \end{align}
\end{subequations}
for $\mA = 1,...,N-1$. We note that the combination of the mixture mass balance \eqref{eq: eq: model alt: mass full} and the $N-1$ constituent balance laws \eqref{eq: eq: model alt: mass} are equivalent to the $N$ balance laws \eqref{eq: eq: model simplified: mass} or $N$ balance laws \eqref{eq: eq: model alt: mass}.

We observe that, in the single-fluid regime, the model simplifies to the compressible Navier–Stokes equations.
\begin{proposition}[Reduction to compressible Navier-Stokes]\label{prop: NS}
If the chemical potentials $\hat{\mu}_\mA$ are well-defined for $\trho_\mA = 0$ ($\mA = 1,...,N$), the multi-constituent system \eqref{eq: model alt} reduces to the compressible Navier-Stokes equations in the single-constituent regime ($Y_\mB= 1$):
\begin{subequations}\label{eq: NS}
  \begin{align} 
 \partial_t (\rho \bv) + {\rm div} \left( \rho \bv\otimes \bv \right)  + \nabla p
 - {\rm div} \left(\nu \left(2\mathbf{D} + \lambda {\rm div}\mathbf{v}\right) \right) -\rho\mathbf{b}&=~ 0, \\
 \partial_t \rho + {\rm div}(\rho \bv) &=~ 0,
  \end{align}
\end{subequations}
with $\rho_\mB = \trho_\mB = \rho$ and $p=p(\rho)$.
\end{proposition}

\subsection{Decomposition of the free energy}\label{subsec: decomp free energy}
We proceed with decomposing the free energy into its components; in agreement with the first metaphysical principle of continuum mixture theory we have:
\begin{align}\label{eq: split free energy}
    \hat{\Psi}\left(\left\{\trho_\mA\right\}_{\mA=1,...,N},\left\{\nabla \trho_\mA\right\}_{\mA=1,...,N}\right) = \sum_\mB \hat{\Psi}_\mB\left(\left\{\trho_\mA\right\}_{\mA=1,...,N},\left\{\nabla \trho_\mA\right\}_{\mA=1,...,N}\right), 
\end{align}
where $\hat{\Psi}_\mB$ is the volume-measure constituent free energy. The split of the free energy, eq. \eqref{eq: split free energy}, highlights that the constituent free energies may depend on quantities of other constituents ($\trho_\mA$ and $\nabla \trho_\mA$ for all $\mA = 1,...,N$) allowing for the incorporation of e.g. attraction/repelling forces.

Next, we associate \textit{constituent chemical potentials} and \textit{constituent partial pressures} with each $\hat{\Psi}_\mB$ in the sense:
\begin{subequations}
    \begin{align}
      \hat{\mu}_{\mB\mA} :=&~ \dfrac{ \partial \hat{\Psi}_\mB}{\partial \trho_\mA} - {\rm div}\dfrac{\partial \hat{\Psi}_\mB}{\partial\nabla \trho_\mA},\\
      p_\mB = &~ \left(\sum_\mA\hat{\mu}_{\mB\mA} \trho_\mA\right) - \hat{\Psi}_\mB.
\end{align}
\end{subequations}
%\newpage
The constituent partial pressure $p_\mA$ are an extension of the classical partial pressure that includes gradient contributions essential to modeling surface tension effects. These constituent quantities satisfy the properties:
\begin{subequations}
  \begin{align}
    \hat{\mu}_\mA =&~ \sum_{\mB} \hat{\mu}_{\mB\mA},\\
    p     = &~ \sum_\mB p_\mB.
  \end{align}
\end{subequations}
The latter may be referred to as \textit{Dalton's law}. 
As such, the split \eqref{eq: split free energy} reveals that the system may be written as:
\begin{subequations}\label{eq: model pressure}
  \begin{align}
   \sum_\mB \left( \partial_t (\trho_\mB \bv) + {\rm div} \left( \trho_\mB \bv\otimes \bv \right) +  \nabla p_\mB + {\rm div} \left(\sum_\mA \nabla \trho_\mA \otimes \dfrac{\partial \hat{\Psi}_\mB}{\partial \nabla \trho_\mA} \right)-\trho_\mB \mathbf{g} \right) & \nn\\
    - {\rm div} \left(   \nu (2\mathbf{D}+\lambda({\rm div}\bv) \mathbf{I}) \right) &=~ 0, \label{eq: model  simplified: mom}\\
  \partial_t \trho_\mA  + {\rm div}(\trho_\mA  \bv) - \sum_{\mB} {\rm div} \left(M_{\mA\mB}\nabla\hat{\mu}_\mB\right)  + \sum_{\mB} m_{\mA\mB} \hat{\mu}_\mB &=~0.\label{eq: eq: model simplified: mass 0}%\\
  %\hat{\mu}_\mA - \dfrac{\partial \Psi_\mA}{\partial \trho_\mA} + {\rm div} \left(  \dfrac{\partial \Psi_\mA}{\partial \nabla \trho_\mA} \right)&=~0,
  \end{align}
\end{subequations}
\begin{remark}[Absence of interaction forces]\label[remark]{rmk: absence interaction forces}
  When interaction forces are absent, i.e. $\hat{\Psi}_\mB = \hat{\Psi}_\mB(\trho_\mB, \nabla \trho_\mB)$ we have
\begin{subequations}
    \begin{align}
    \hat{\mu}_{\mB\mA}=&~\hat{\mu}_\mA \delta_{\mA\mB},\\
    p_{\mB}=&~ \hat{\mu}_\mB\trho_\mB - \hat{\Psi}_\mB,
\end{align}
\end{subequations}
so that
\begin{align}
    \sum_\mA \nabla \trho_\mA \otimes \dfrac{\partial \hat{\Psi}_\mB}{\partial \nabla \trho_\mA} = \nabla \trho_\mB \otimes \dfrac{\partial \hat{\Psi}_\mB}{\partial \nabla \trho_\mB}.
\end{align}

\end{remark}

\subsection{Equilibrium conditions}\label{subsec: equilibrium conditions}
We characterize the equilibrium conditions of the model \eqref{eq: model pressure} by:
\begin{align}
 \mathscr{D} = 0.
\end{align}
The identity \eqref{eq: second law red 7} provides:
\begin{subequations}\label{eq: static eqs}
\begin{align}
     2 \nu^{\rm eq} \left( \mathbf{D}^{\rm eq} - \frac{1}{d} ({\rm div} \mathbf{v}^{\rm eq}) \mathbf{I}\right):\left(\mathbf{D}^{\rm eq} - \frac{1}{d} ({\rm div} \mathbf{v}^{\rm eq}) \mathbf{I}\right)=&~0,\label{eq: static eqs 1}\\ 
     \nu^{\rm eq} \left(\lambda + \frac{2}{d}\right)\left({\rm div} \mathbf{v}^{\rm eq}\right)^2 =&~0,\label{eq: static eqs 2}\\
     \displaystyle\sum_{\mA,\mB} \mathbf{M}_{\mA\mB}^{\rm eq}\nabla \hat{\mu}_\mA^{\rm eq}\cdot\nabla \hat{\mu}_\mB^{\rm eq} =&~0,\label{eq: static eqs 3}\\
     \sum_{\mA,\mB} m_{\mA\mB}^{\rm eq}\hat{\mu}_\mA^{\rm eq}\hat{\mu}_\mB^{\rm eq} =&~0,\label{eq: static eqs 4}
\end{align}
\end{subequations}
where the superscripts `${\rm eq}$' denote the equilibrium of the quantity.
Identities \eqref{eq: static eqs 1} and \eqref{eq: static eqs 2} dictate that $\mathbf{v}^{\rm eq}$ are rigid motions. Since $\mathbf{M}_{\mA\mB}$ is symmetric positive definite, we decompose it as $(\mathbf{M}_{ij})_{\mA\mB}=(\mathbf{B}_{ij})_{\mA}(\mathbf{B}_{ij})_{\mB}$ (with $i,j$ the indices of the tensor), for some $\mathbf{B}_\mA$ with the same dependencies as $\mathbf{M}_{\mA\mB}$. Hence, \eqref{eq: static eqs 3} implies $\sum_\mA \mathbf{B}_\mA^{\rm eq} \nabla \hat{\mu}_\mA^{\rm eq} =0$. In a similar fashion we deduce $\sum_\mA b_\mA^{\rm eq} \hat{\mu}_\mA^{\rm eq} =0$ where $m_{\mA\mB} = b_{\mA}b_{\mB}$. These conditions imply $\hat{\bH}^{\rm eq}_\mA = 0, \hat{\zeta}^{\rm eq}_\mA=0$ for $\mA=1,...,N$. Finally, restricting to zero velocities, $\mathbf{v}^{\rm eq}=\mathbf{0}$, the mass balance laws imply constant in time partial densities, $\trho_\mA^{\rm eq}(\mathbf{x},t) = \trho_\mA^{\rm eq}(\mathbf{x})$, while the momentum balance dictates: $\sum_\mA Y_\mA^{\rm eq}\nabla \hat{\mu}_\mA^{\rm eq} = \mathbf{g}$. Summarizing, we have the equilibrium conditions:
\begin{subequations}\label{eq: equilibrium chem}
    \begin{align}
      \mathbf{v}^{\rm eq}=&~\mathbf{0},\\
      \trho_\mA^{\rm eq}=&~{\rm const }~ (\text{in time}),\\
      \sum_\mA b_\mA^{\rm eq} \hat{\mu}_\mA^{\rm eq} =&~0,\\
      \sum_\mA \mathbf{B}_\mA^{\rm eq} \nabla \hat{\mu}_\mA^{\rm eq} =&~0,\\
      \sum_\mA Y_\mA^{\rm eq} \nabla \hat{\mu}_\mA^{\rm eq} =&~\mathbf{g}.\label{eq: equilibrium chem last}
\end{align}
\end{subequations}
\begin{remark}[No equilibrium conditions pressure]
    If the free energy $\Psi$ does not depend on $\left\{\nabla \trho_\mA\right\}$, the condition \eqref{eq: equilibrium chem last} converts into $\nabla p = \rho \mathbf{g}$ (see \cref{lem: identity Korteweg stresses}). Thus there is no equilibrium condition on the partial pressures $p_\mA, \mA=1,...,N$. This is in contrast to some compressible two-phase models in which pressure equilibrium is assumed, see e.g. \cite{kapila2001two,murrone2005five,daude2023hyperbolic} (and \cref{subsec: 1st order system}).
\end{remark}

\section{Analysis of the first-order system}\label{sec: hyperbolicity}
In this section, we examine the first-order system associated with \eqref{eq: model full}. We first introduce the system in \cref{subsec: 1st order}, then analyze its potential hyperbolic structure in \cref{subsec: hyperbolic}, and finally derive the Riemann invariants and jump conditions in \cref{subsec: Riemann}.

\subsection{First-order system}\label{subsec: 1st order}
We study the first-order system associated with \eqref{eq: model full} in absence of mass transfer ($\hat{\zeta}_\mA = 0$) and body forces ($\mathbf{b}=0$):
\begin{subequations}\label{eq: model 1st order}
  \begin{align}
   \partial_t (\rho \bv) + {\rm div} \left( \rho \bv\otimes \bv \right) + \nabla p &=~ 0, \label{eq: model  simplified: mom 2}\\
  \partial_t \trho_\mA  + {\rm div}(\trho_\mA  \bv) &=~0.\label{eq: eq: model simplified: mass 2}
  \end{align}
\end{subequations}
The system satisfies the energy law \eqref{eq: energy dissipation} with $\mathscr{D}=0$. This may be written in the local form:
\begin{align}\label{eq: local form energy law 0}
    \partial_t \left(\Psi + \bar{\mathscr{K}}\right) + {\rm div}\left(\left(\Psi + \bar{\mathscr{K}}\right)\bv\right) = 0.
\end{align}
\begin{remark}[Derivation energy law]
  The energy law \eqref{eq: local form energy law 0} follows directly from \eqref{eq: model 1st order} by multiplying with the appropriate weights. In this regard, we notice the following identities:
  \begin{subequations}
      \begin{align}
         \sum_{\mA} \dfrac{\partial \Psi}{\partial \trho_\mA} \left( \partial_t \trho_\mA  + {\rm div}(\trho_\mA  \bv) \right) =&~ \partial_t \Psi + {\rm div} (\Psi \mathbf{v}) + p {\rm div} \bv,\\
         \sum_{\mA} \left(-\frac{1}{2}|\bv|^2\left( \partial_t \trho_\mA  + {\rm div}(\trho_\mA  \bv) \right)\right) =&~ -\frac{1}{2}|\bv|^2\left( \partial_t \rho + {\rm div} (\rho \mathbf{v})\right),\\  
         \bv\cdot \left( \partial_t (\rho \bv)  + {\rm div}(\rho  \bv \otimes \bv) + \nabla p \right) =&~ \partial_t \bar{\mathscr{K}} + {\rm div}(\bar{\mathscr{K}} \bv) + \bv\cdot\nabla  p \nn\\
         &~+ \frac{1}{2}|\bv|^2 \left(\partial_t \rho + {\rm div} (\rho \mathbf{v})\right).
      \end{align}
  \end{subequations}
\end{remark}
\begin{remark}[Energy-law for non-smooth solutions]
  In case of shocks or other regularities the stability condition takes the form:
\begin{align}\label{eq: local form energy law 1}
    \partial_t \left(\Psi + \bar{\mathscr{K}}\right) + {\rm div}\left(\left(\Psi + \bar{\mathscr{K}}\right)\bv\right) \leq 0,
\end{align}
  which may be established using vanishing-viscosity solutions.
\end{remark}
In order to analyze the wave speeds of the model, we deduce the evolution equations of the partial pressures from \eqref{eq: eq: model simplified: mass}:
%For smooth solutions, the evolution equations take the form:
\begin{subequations}\label{eq: system prim 0}
\begin{align}
   \partial_t p_\mA + \bv \cdot \nabla p_\mA + \sum_{\mB} \trho_\mB a_{\mA\mB}^2 {\rm div} \bv =0,
\end{align}
\end{subequations}
where the speed of sound quantities $a_{\mA\mB}$ account for the interaction between constituents:
\begin{align}
    a_{\mA\mB}^2 := \dfrac{\partial p_\mA}{\partial \trho_\mB}.% = \sum_\mC \trho_\mC \dfrac{\partial^2 \hat{\Psi}_\mA}{\partial \trho_\mB \partial \trho_\mC}.
\end{align}
The evolution of the classical mixture pressure results from the superposition of the constituent pressure equations:
\begin{align}
    \partial_t p + \bv \cdot \nabla p + \rho a^2 {\rm div} \bv= 0,
\end{align}
where the mixture speed of sound $a$ satisfies:
\begin{align}\label{eq: mix speed of sound}
    \rho a^2 =&~ \sum_{\mB} \trho_\mB a_{\mB}^2,
\end{align}   
and the constituent speed of sound quantities are given by:
\begin{align}\label{eq: const speed of sound}
    a_\mB^2 =&~  \dfrac{\partial p}{\partial \trho_\mB}= \sum_{\mA} a_{\mA\mB}^2.
\end{align}   
\begin{remark}[Speed of sound compressible two-phase models]
The explicit dependence of the free energy on each constituent (see eq. \eqref{eq: class Psi}) is reflected in the definitions of the sound speeds \(a_\mB\) and \(a_{\mA\mB}\). This approach contrasts with the sound speed definitions typically employed in compressible two-phase flow models (e.g., \cite{baer1986two}). In the absence of interaction forces (cf. \cref{rmk: absence interaction forces}), the relation reduces to $a_\mB^2 = \mathrm{d}p_\mB/\mathrm{d}\trho_\mB$, where the standard definition in those models is $a_\mB^2 = \mathrm{d}p_\mB/\mathrm{d}\rho_\mB$.
\end{remark}

\subsection{Hyperbolic structure}\label{subsec: hyperbolic}
In the following we analyze the hyperbolic structure of the first-order model
\begin{subequations}\label{eq: system prim}
\begin{align}
    \partial_t \bv + \bv \cdot \nabla \bv + \rho^{-1}\nabla p =&~0,\\
    \partial_t p_\mA + \bv \cdot \nabla p_\mA + \sum_{\mB} \trho_\mB a_{\mA\mB}^2 {\rm div} \bv =&~0,
\end{align}
\end{subequations}
where the momentum equation is a direct consequence from \eqref{eq: model 1st order}. In this section we consider the one-dimensional case $d=1$ for conciseness, and refer to \cref{app: hyperbolic} for dimension $d=3$. We consider the case $a_{\mA\mB},a_{\mB},a>0$. We write the system \eqref{eq: system prim} in the matrix-vector form:
\begin{align}
  \partial_t \mathbf{W} + \mathbf{A} \partial_x \mathbf{W}  = 0,
\end{align}
where
\begin{align}
  \mathbf{W} = \begin{bmatrix}
      p_1 \\
      \vdots \\
      p_N\\
      v 
  \end{bmatrix}, \quad  \mathbf{A} =\begin{bmatrix}
    v           & 0           & \cdots  &        & 0           & \sum_\mB \trho_\mB a_{1\mB}^2\\
    0           & v           & 0      & \cdots  & 0           & \sum_\mB \trho_\mB a_{2\mB}^2\\
    \vdots      &  \ddots           & \ddots & \ddots       &  \vdots           & \vdots\\
    0           &   \cdots    &  0      & v      &      0      & \sum_\mB \trho_\mB a_{(N-1)\mB}^2\\
    0           &   \cdots    &        & 0      &      v      & \sum_\mB \trho_\mB a_{N\mB}^2\\
    \rho^{-1} &      & \cdots  &        & \rho^{-1}   & v  
  \end{bmatrix}.
\end{align}
The system admits $N+1$ eigenvalues:
\begin{align}
    \lambda_1 = v - a < \lambda_2 = ... = \lambda_{N} = v < \lambda_{N+1} = v+a.
\end{align}
The corresponding right eigenvectors are:
\begin{align}
    \br_1 = \begin{bmatrix}
        - \sum_\mB \trho_\mB a_{1\mB}^2\\
        \vdots \\
        - \sum_\mB \trho_\mB a_{N\mB}^2\\
        a
    \end{bmatrix}, ~
    \br_2 = \begin{bmatrix}
        - 1\\
        1 \\
        0\\
        \vdots \\
        \\
        0
    \end{bmatrix}, ~
    \br_3 = \begin{bmatrix}
        - 1\\
        0 \\
        1\\
        0\\
        \vdots \\
        0
    \end{bmatrix}, \cdots ,~
    \br_N = \begin{bmatrix}
        - 1\\
        0 \\
        \vdots \\
        0\\
        1\\
        0
    \end{bmatrix}, ~
    \br_{N+1} =\begin{bmatrix}
        \sum_\mB \trho_\mB a_{1\mB}^2\\
        \vdots \\
        \sum_\mB \trho_\mB a_{N\mB}^2\\
        a
    \end{bmatrix}, 
\end{align}
and the left eigenvectors are:
\begin{align}
    \bl_1 = \begin{bmatrix}
        - \dfrac{1}{2 \rho a^2}\\
        \vdots \\
        - \dfrac{1}{2 \rho a^2}\\
         \dfrac{1}{2 a}
    \end{bmatrix}, ~
    \bl_2 = \begin{bmatrix}
        - \dfrac{\sum_\mB \trho_\mB a_{2\mB}^2}{\sum_\mB \trho_\mB (a_{1\mB}^2+a_{2\mB}^2)}\\
          \dfrac{\sum_\mB \trho_\mB a_{1\mB}^2}{\sum_\mB \trho_\mB (a_{1\mB}^2+a_{2\mB}^2)}\\
        0\\
        \vdots\\
        0
    \end{bmatrix}, ~
    \bl_3 = \begin{bmatrix}
        - \dfrac{\sum_\mB \trho_\mB a_{3\mB}^2}{\sum_\mB \trho_\mB (a_{1\mB}^2+a_{3\mB}^2)}\\
        0\\
          \dfrac{\sum_\mB \trho_\mB a_{1\mB}^2}{\sum_\mB \trho_\mB (a_{1\mB}^2+a_{3\mB}^2)}\\
        0\\
        \vdots\\
        0
    \end{bmatrix}, \cdots ,\\~
    \bl_N = \begin{bmatrix}
        - \dfrac{\sum_\mB \trho_\mB a_{N\mB}^2}{\sum_\mB \trho_\mB (a_{1\mB}^2+a_{N\mB}^2)}\\
        0\\
        \vdots \\
        0\\
          \dfrac{\sum_\mB \trho_\mB a_{1\mB}^2}{\sum_\mB \trho_\mB (a_{1\mB}^2+a_{N\mB}^2)}\\
        0
    \end{bmatrix}, ~
    \bl_{N+1} =\begin{bmatrix}
         \dfrac{1}{2 \rho a^2}\\
        \vdots \\
         \dfrac{1}{2 \rho a^2}\\
         \dfrac{1}{2 a}
    \end{bmatrix}. 
\end{align}
The eigenvectors are scaled so that $\br_j\cdot\bl_j=1$ for all $j=1,...,N+1$. Both right and left eigenvectors are linearly independent and span $\mathbb{R}^{N+1}$. This implies that the system is  hyperbolic if the eigenvalues are real-valued, which holds when $a_{\mA\mB}^2 >0$ for all $\mA,\mB = 1,...,N$.
\begin{remark}[Van der Waals equation of state]
    The van der Waals equation of state does not, in general, guarantee that \(a_{\mA\mB}^2 > 0\), and consequently, the resulting system is not hyperbolic.
\end{remark}

\begin{proposition}[Structure of the waves]\label{prop: waves}
The characteristic fields associated with eigenvalues $\lambda_j, j=2,...,N$ are linearly degenerate (i.e. $\mathbf{r}_j\cdot\partial \lambda_j/\partial \mathbf{W}=0$), and those associated with eigenvalues $\lambda_j, j=1,N+1$ are genuinely non-linear (i.e. $\mathbf{r}_j\cdot\partial \lambda_j/\partial \mathbf{W}\neq 0$).
\end{proposition}

This means that fields associated with eigenvalues $\lambda_j, j=2,...,N$ correspond to contact discontinuities, whereas fields associated with eigenvalues $\lambda_j, j=1,N+1$ correspond to and shock or rarefaction waves, respectively.

\begin{proof}
A direct computation shows that the fields associated with eigenvalues $\lambda_j, j=2,...,N$ are linearly degenerate (i.e. $\mathbf{r}_j\cdot\partial \lambda_j/\partial \mathbf{W}=0$). 

Next we focus on the field associated with eigenvalues $\lambda_j, j=1,N+1$. We have:
\begin{align}\label{eq: wave 1, N+1}
  \mathbf{r}_1\cdot\partial \lambda_1/\partial \mathbf{W}=\mathbf{r}_{N+1}\cdot\partial \lambda_{N+1}/\partial \mathbf{W}=  a+\sum_{\mA\mB} \dfrac{\partial a}{\partial p_\mA}\trho_\mB a_{\mA\mB}^2.
\end{align}
Next we note that for an arbitrary variable $\omega$ we have:
\begin{align}
  \dfrac{\partial a}{\partial \omega} = \frac{1}{2a \rho}\left(\dfrac{\partial (\rho a^2)}{\partial \omega}-a^2 \frac{\partial \rho}{\partial \omega}\right).
\end{align}
Applying this identity for $\omega = p_\mA$ and substituting \eqref{eq: mix speed of sound} yields:
\begin{align}
  \dfrac{\partial a}{\partial p_\mA} =&~ \frac{1}{2a \rho}\left(\sum_{\mC\mT} \dfrac{\partial (\trho_\mT a_{\mC\mT}^2)}{\partial p_\mA}-a^2 \sum_\mT \frac{\partial \trho_\mT}{\partial p_\mA}\right)\nn\\
  %=&~ \frac{1}{2a \rho}\left(\sum_{\mT\mC} \left(\trho_\mT \dfrac{\partial  a_{\mC\mT}^2}{\partial p_\mA}+a_{\mC\mT}^2\dfrac{\partial \trho_\mT }{\partial p_\mA}\right)-a^2 \sum_\mT\frac{\partial \trho_\mT}{\partial p_\mA}\right)\nn\\
  =&~ \frac{1}{2a \rho}\left(\sum_{\mT\mC} \trho_\mT \dfrac{\partial  a_{\mC\mT}^2}{\partial p_\mA}+\sum_{\mT\mC} a_{\mC\mT}^2(a^2)^{-1}_{\mT\mA}-a^2 \sum_\mT (a^2)^{-1}_{\mT\mA}\right)\nn\\
  =&~ \frac{1}{2a \rho}\left(\sum_{\mT} \trho_\mT \dfrac{\partial  a_{\mT}^2}{\partial p_\mA}+1-a^2 \sum_{\mT}(a^2)^{-1}_{\mT\mA}\right).
\end{align}
Inserting into \eqref{eq: wave 1, N+1} gives:
\begin{align}\label{eq: wave 1, N+1 2}
  \mathbf{r}_1\cdot\partial \lambda_1/\partial \mathbf{W}=\mathbf{r}_{N+1}\cdot\partial \lambda_{N+1}/\partial \mathbf{W}=&~  a+\sum_{\mA\mB} \dfrac{\partial a}{\partial p_\mA}\trho_\mB a_{\mA\mB}^2\nn\\
  =&~  a+\frac{1}{2a \rho}\sum_{\mA\mB} \trho_\mB a_{\mA\mB}^2 \left(\sum_{\mT} \trho_\mT \dfrac{\partial  a_{\mT}^2}{\partial p_\mA}+1-a^2 \sum_{\mT}(a^2)^{-1}_{\mT\mA}\right) \nn\\
  %=&~  a+\frac{1}{2a \rho}\left(\sum_{\mA\mB} \trho_\mB a_{\mA\mB}^2  - a^2 \sum_{\mA\mB\mT}\trho_\mB  a_{\mA\mB}^2(a^2)^{-1}_{\mT\mA}\right) \nn\\
  %=&~  a+\frac{1}{2a \rho}\left(\sum_{\mA\mB} \trho_\mB a_{\mA\mB}^2  - a^2 \sum_{\mB\mT}\trho_\mB  \delta_{\mT\mB}\right) \nn\\
  =&~  a+\frac{1}{2a \rho}\left(\sum_{\mA\mB\mT} \trho_\mB a_{\mA\mB}^2  \trho_\mT \dfrac{\partial  a_{\mT}^2}{\partial p_\mA}+\sum_{\mA\mB} \trho_\mB a_{\mA\mB}^2  - a^2 \sum_{\mT}\trho_\mT\right) \nn\\
  =&~  a +\frac{1}{2a \rho}\sum_{\mA\mB\mT} \trho_\mB a_{\mA\mB}^2  \trho_\mT \dfrac{\partial  a_{\mT}^2}{\partial p_\mA},
\end{align}
which is strictly positive by the aforementioned assumption $a_{\mA\mB},a>0$ in combination with noting that the phase speed of sound increases with the pressure $\partial a_{\mT}/\partial p_\mA>0$ (see also \cite{murrone2005five} for a similar assumption).

%Applying this identity for $\omega = p_\mA$ and substituting \eqref{eq: mix speed of sound} yields:
%\begin{align}
%  \dfrac{\partial a}{\partial p_\mA} =&~ \frac{1}{2a \rho}\left(\sum_{\mB} \dfrac{\partial (\trho_\mB a_{\mB}^2)}{\partial p_\mA}-a^2 \sum_\mB \frac{\partial \trho_\mB}{\partial p_\mA}\right)\nn\\
%  =&~ \frac{1}{2a \rho}\left(\sum_{\mB} \left(\trho_\mB \dfrac{\partial  a_{\mB}^2}{\partial p_\mA}+a_{\mB}^2\dfrac{\partial \trho_\mB }{\partial p_\mA}\right)-a^2 \sum_\mB\frac{\partial \trho_\mB}{\partial p_\mA}\right)\nn\\
%  =&~ \frac{1}{2a \rho}\left(\sum_{\mB} a_{\mB}^2\dfrac{\delta_{\mA\mB}}{a_{\mA}^2}-a^2 \sum_\mB\dfrac{\delta_{\mA\mB}}{a_\mA^2}\right)\nn\\
%  =&~ \frac{1}{2a \rho}\left(1- \dfrac{a^2}{a_\mA^2}\right).
%\end{align}
%Inserting into \eqref{eq: wave 1, N+1} gives:
%\begin{align}\label{eq: wave 1, N+1}
%  \mathbf{r}_1\cdot\partial \lambda_1/\partial \mathbf{W}=\mathbf{r}_{N+1}\cdot\partial \lambda_{N+1}/\partial \mathbf{W}=  a+\frac{1}{2a \rho}\sum_\mA \trho_\mA\left( a_\mA^2- a^2\right) = a.
%\end{align}
 
\end{proof} 

\subsection{Riemann invariants and jump conditions}\label{subsec: Riemann}
Next, we derive the Riemann invariants of the first-order system \eqref{eq: model 1st order}. 
\begin{proposition}\label{prop: Riemann invariants}
The Riemann invariants associated with the various waves are: 
\begin{subequations}\label{eq: Riemann invariants}
  \begin{align}
     \lambda_1: & \quad \left\{\left\{Y_\mA\right\}_{\mA\neq \mB}, v+\int_p \frac{1}{\rho a} {\rm d}p\right\} \\
     \lambda_2,...,\lambda_{N}: & \quad \left\{v,p\right\}\\
     \lambda_{N+1}: &  \quad \left\{\left\{Y_\mA\right\}_{\mA\neq \mB}, v-\int_p \frac{1}{\rho a} {\rm d}p \right\},
  \end{align}
\end{subequations}
for some fixed $\mB \in \left\{1,...,N\right\}$.
\end{proposition}
\begin{proof}
To find the Riemann invariants, we seek for $\omega_i$ such that $\partial_{\mathbf{W}} \omega_i \cdot \mathbf{r}_j = 0$. Focusing on the Riemann invariants associated with $\lambda_1$, we have:
\begin{align}
   0= \dfrac{\partial \omega_i}{\partial \mathbf{W}}\cdot \mathbf{r}_1 =  a \dfrac{\partial \omega_i}{\partial v} - \sum_\mB \trho_\mB a_\mB^2 \dfrac{\partial \omega_i}{\partial p_\mB}.
  \end{align}
A direct computation reveals that $Y_\mA = \trho_\mA/\rho$ are Riemann invariants:
\begin{align}
   \dfrac{\partial Y_\mA}{\partial \mathbf{W}}\cdot \mathbf{r}_1 =&~ - \sum_\mB \trho_\mB a_\mB^2 \dfrac{\partial Y_\mA}{\partial p_\mB}\nn\\
    =&~ -\frac{1}{\rho} \sum_\mB \trho_\mB a_\mB^2\dfrac{\partial \trho_\mA}{\partial p_\mB} + \frac{\trho_\mA}{\rho^2}\sum_\mB \trho_\mB a_\mB^2  \dfrac{\partial \rho}{\partial p_\mB}\nn\\
    =&~ -\frac{1}{\rho} \trho_\mA + \frac{\trho_\mA}{\rho^2}\sum_\mB \trho_\mB   \nn\\
    =&~ 0.
  \end{align}
Next, we verify that $v+\int_p \frac{1}{\rho a} {\rm d}p$ is the final Riemann invariant:
\begin{align}
   \dfrac{\partial v+\int_p \frac{1}{\rho a} {\rm d}p}{\partial \mathbf{W}}\cdot \mathbf{r}_1 =&~  a \dfrac{\partial v+\int_p \frac{1}{\rho a} {\rm d}p}{\partial v} - \sum_\mA \trho_\mA a_\mA^2 \dfrac{\partial v+\int_p \frac{1}{\rho a} {\rm d}p}{\partial p_\mA}\nn\\
    =&~a -\sum_\mA \trho_\mA a_\mA^2\dfrac{\partial \int_p \frac{1}{\rho a} {\rm d}p}{\partial p_\mA} \nn\\
    =&~ a-\sum_\mA \trho_\mA a_\mA^2\frac{1}{\rho a}   \nn\\
    =&~0.
  \end{align}
In a similar fashion we find that the Riemann invariants associated with $\lambda_{N+1}$ are $Y_\mA$, $\mA = 1,...,N$ and $v-\int_p \frac{1}{\rho a} {\rm d}p$.
Finally, we compute the Riemann invariants associated with $\lambda_2,...,\lambda_N$:
  \begin{align}
   0= \dfrac{\partial \omega_i}{\partial \mathbf{W}}\cdot \mathbf{r}_j = - \dfrac{\partial \omega_i}{\partial p_1} + \dfrac{\partial \omega_i}{\partial p_j},\quad \text{for }j = 2,...,N.
  \end{align}
  This is equivalent to
    \begin{align}\label{eq: system RI}
   \dfrac{\partial \omega_i}{\partial p_1} = \dots = \dfrac{\partial \omega_i}{\partial p_N}.
  \end{align}
It is now easy to verify that $\omega_i = v,p$ solve system \eqref{eq: system RI}.
\end{proof}
Finally, we note that the conservative nature of the system permits to uniquely define shock-waves due to the Rankine-Hugoniot conditions:
\begin{subequations}\label{eq: Rankine-Hugoniot conditions}
    \begin{align}
        \Delta \left(\trho_\mA (v-\sigma)\right) =&~ 0, \quad\mA = 1,...,N\\
        \Delta \left(\rho v(v-\sigma) + p\right) =&~ 0,
    \end{align}
\end{subequations}
with $\Delta \phi = \phi_R - \phi_L$ the jump between the left and right states of the shock-wave, and $\sigma$ the speed of the shock-wave.

\section{Binary mixtures}\label{sec: binary mixtures}

In this section we restrict to binary mixtures ($\mA = 1,2)$
%\subsection{Alternative formulations}
for which we discuss alternative -- equivalent -- forms. The model \eqref{eq: model simplified} takes for binary mixtures the following form:
\begin{subequations}\label{eq: model binary}
  \begin{align}
   \partial_t (\rho \bv) + {\rm div} \left( \rho \bv\otimes \bv \right) + \trho_1\nabla\hat{\mu}_1 + \trho_2\nabla\hat{\mu}_2
    - {\rm div} \left(   \nu (2\mathbf{D}+\lambda({\rm div}\bv) \mathbf{I}) \right)-\rho\mathbf{g} &=~ 0, \label{eq: model binary: mom}\\
  \partial_t \trho_1  + {\rm div}(\trho_1  \bv) - {\rm div} \left(M\nabla\left(\hat{\mu}_1-\hat{\mu}_2)\right)\right)  + m \left(\hat{\mu}_1-\hat{\mu}_2)\right) &=~0,\label{eq: model binary: mass 1}\\
  \partial_t \trho_2  + {\rm div}(\trho_2  \bv) - {\rm div} \left(M\nabla\left(\hat{\mu}_2-\hat{\mu}_1)\right)\right)  + m \left(\hat{\mu}_2-\hat{\mu}_1)\right) &=~0,\label{eq: model binary: mass 2}
  \end{align}
\end{subequations}
where $M =- M_{12}=-M_{21}=M_{11}=M_{22}$ and $m=-m_{12}=-m_{21}=m_{11}=m_{22}$. By means of variable transformation, we may express the model in terms of the standard mixture quantities. The diffusive fluxes and mass transfer terms constitute a single quantity by means of the balance conditions \eqref{eq: rel gross motion zero} and \eqref{eq: balance mass fluxes}:
\begin{subequations}
    \begin{align}
        \hat{\zeta} :=&~ \hat{\zeta}_1-\hat{\zeta}_2, \quad\quad \hat{\zeta}_1 = \frac{1}{2}\hat{\zeta}, \quad\quad \hat{\zeta}_2= -\frac{1}{2}\hat{\zeta},\\
        \hat{\bH} :=&~ \hat{\bH}_1-\hat{\bH}_2, \quad\quad \hat{\bH}_1 = \frac{1}{2}\hat{\bH}, \quad\quad \hat{\bH}_2= -\frac{1}{2}\hat{\bH}.
    \end{align}
\end{subequations}
Additionally, we introduce a mass-fraction order parameter:
\begin{align}\label{eq: order parameters}
  Y :=&~ Y_1-Y_2, \quad\quad    Y_1=\frac{1+Y}{2}, \quad\quad Y_2=\frac{1-Y}{2},
\end{align}
and deduce the relation for the density of the mixture $\rho$:
\begin{subequations}\label{eq: var trans}
\begin{align}
    \frac{2}{\rho} =\frac{Y_1}{\trho_1} + \frac{Y_2}{\trho_2}=\frac{1}{\trho_1}\frac{1+Y}{2} + \frac{1}{\trho_2}\frac{1-Y}{2},\\
    \trho_1 = \rho \frac{1+Y}{2}, \quad \trho_2 =\rho \frac{1-Y}{2}, \quad \rho  = \trho_1+\trho_2, \quad  Y = \dfrac{\trho_1-\trho_2}{\trho_1+\trho_2}.
\end{align}
\end{subequations}
Furthermore, expressing the model solely in terms of mixture quantities requires the conversion of the component chemical potential quantities. To this purpose we define:
\begin{subequations}\label{eq: chem pot rho c Psi}
\begin{align}
    \hat{\mu}_\rho =&~ \dfrac{ \partial \hat{\Psi}}{\partial \rho} - {\rm div}\dfrac{\partial \hat{\Psi}}{\partial\nabla \rho},\\
    \hat{\mu}_Y =&~ \dfrac{ \partial \hat{\Psi}}{\partial Y} - {\rm div}\dfrac{\partial \hat{\Psi}}{\partial\nabla Y}.
\end{align}
\end{subequations}  %\newpage
Additionally, we introduce the following identities.
\begin{lemma}[Transformation identities]\label[lemma]{lem: mix identities}
The component chemical potential quantities may be expressed in mixture quantities by means of the identities: 
\begin{subequations} \begin{align}
     \trho_1 \nabla \hat{\mu}_1+\trho_2\nabla\hat{\mu}_2=&~\rho \nabla \hat{\mu}_\rho - \hat{\mu}_Y \nabla Y,\\
      \hat{\mu}_1-\hat{\mu}_2 =&~ 2\frac{\hat{\mu}_Y}{\rho}.
  \end{align}
\end{subequations}
\end{lemma}
\begin{proof}
See \cref{lem: appendix: relation chem}.
\end{proof}
Applying the variable transformation \eqref{eq: var trans} and the identities of \cref{lem: mix identities} provides:
\begin{subequations}\label{eq: model simplified equivalent}
  \begin{align}
  \partial_t (\rho \bv) + {\rm div} \left( \rho \bv\otimes \bv \right) + \rho \nabla \hat{\mu}_\rho - \hat{\mu}_Y \nabla Y
    - {\rm div} \left(   \nu (2\mathbf{D}+\lambda({\rm div}\bv) \mathbf{I}) \right)-\rho\mathbf{g} &=~ 0, \label{eq: model binary: mom 1}\\
  \partial_t \rho + {\rm div}(\rho \bv) &=~ 0,\label{eq: model simplified equivalent: rho}\\
 \partial_t (\rho Y) + {\rm div}(\rho Y \bv) - {\rm div} \left(\hat{M}\nabla\frac{\hat{\mu}_Y}{\rho}\right)  + \hat{m} \frac{\hat{\mu}_Y}{\rho} &=~0, \label{eq: model simplified equivalent: rho c}
 \end{align}
\end{subequations} 
where $\hat{M}=4M$ and $\hat{m}=4m$. Identifying
$ \hat{\Psi}(\trho_1,\nabla \trho_1,\trho_2,\nabla \trho_2) \equiv \check{\Psi}(\rho,\nabla \rho, Y,\nabla Y) $ so that $ \hat{\mu}\rho = \check{\mu}\rho $ and $ \hat{\mu}_Y = \check{\mu}_Y $, 
and setting $ \hat{M} = \check{M} $,reveals that \eqref{eq: model simplified equivalent} matches with the models derived in \cref{sec: appendix: alternative constitutive modeling: rho c Psi}.

Furthermore, the model may be expressed using a mass-measured free energy $\hat{\psi} = \rho^{-1}\hat{\Psi}$. For this purpose we introduce the chemical potentials:
\begin{subequations}
\begin{align}
    \hat{\nu}_\rho =&~ \dfrac{ \partial \hat{\psi}}{\partial \rho} - \frac{1}{\rho}{\rm div}\left(\rho\dfrac{\partial \hat{\psi}}{\partial\nabla \rho}\right),\\
    \hat{\nu}_Y =&~ \dfrac{ \partial \hat{\psi}}{\partial Y} - \frac{1}{\rho}{\rm div}\left(\rho \dfrac{\partial \hat{\Psi}}{\partial\nabla Y}\right),
\end{align}
\end{subequations}
which are related to \eqref{eq: chem pot rho c Psi} via $\hat{\mu}_\rho = \rho \hat{\nu}_\rho + \hat{\psi}$ and $\hat{\mu}_Y = \rho \hat{\nu}_Y$ (see also the proof of \cref{prop: Equivalence binary mixture models}). Inserting these provides the alternative formulation:
\begin{subequations}\label{eq: model simplified equivalent 2}
  \begin{align}
  \partial_t (\rho \bv) + {\rm div} \left( \rho \bv\otimes \bv \right) + \rho \nabla (\rho \hat{\nu}_\rho + \hat{\psi}) - \rho \hat{\nu}_Y \nabla Y&\nn\\
    - {\rm div} \left(   \nu (2\mathbf{D}+\lambda({\rm div}\bv) \mathbf{I}) \right)-\rho\mathbf{g} &=~ 0, \label{eq: model binary: mom 2}\\
  \partial_t \rho + {\rm div}(\rho \bv) &=~ 0,\label{eq: model simplified equivalent: rho 2}\\
 \partial_t (\rho Y) + {\rm div}(\rho Y \bv) - {\rm div} \left(\hat{M}\nabla\hat{\nu}_Y\right)  + \hat{m} \hat{\nu}_Y &=~0. \label{eq: model simplified equivalent: rho c 2}
 \end{align}
\end{subequations} 
This formulation aligns with the model derived in \cref{subsec: mass measure free energy} through the identification
$ \hat{\psi}(\trho_1,\nabla \trho_1,\trho_2,\nabla \trho_2) \equiv \doublecheck{\psi}(\rho,\nabla \rho, Y,\nabla Y) $,
which implies $ \hat{\nu}_\rho = \doublecheck{\mu}_\rho $ and $ \hat{\nu}_Y = \doublecheck{\mu}_Y $, and additionally identifying $ \hat{M} = \doublecheck{M} $.

%Finally, we note that the equilibrium conditions \eqref{eq: equilibrium chem} take the form:
%\begin{subequations}\label{eq: equilibrium chem rho c}
%    \begin{align}
%    \hat{\mu}_Y^{\rm eq}&~=0,\\
%    \hat{\mu}^{\rm eq}_\rho + b y &~= {\rm const}.
%\end{align}
%\end{subequations}

\section{Connections to existing models}\label{sec: Connections} 

This section aims to establish connections with existing models: in \cref{subsec: 1st order system}, we compare with established compressible two-phase models; in \cref{subsec: Binary flows}, we examine links to phase-field models for binary fluids; and in \cref{subsec: Incompressible N-phase flows}, we highlight similarities with an $N$-phase incompressible model.

\subsection{Compressible two-phase models}\label{subsec: 1st order system}

%An important class of models are the compressible two-phase models. These first-order systems are categorized by the number of equations in one spatial dimension for the case of two phases. In this classification, the seven equation model often considered as the most complete model \cite{baer1986two}. 
Motivated by revealing similarities between models, here we describe the connection of the first-order system \eqref{eq: model 1st order} to existing \textit{three-equation} compressible two-phase models.\\

%\underline{\cite{dumbser2011simple}}. 
The model of \cite{dumbser2011simple} describing a two-phase flow scenario with phase 1 a liquid phase and phase 2 a gas phase, is given by\footnote{The model of \cite{dumbser2011simple} contains an additional term that represents gravitational forces. In this comparison we consider the version without this term.}:
\begin{subequations}\label{eq: model Dumbser}
  \begin{align}
   \partial_t (\rho_1\phi_1 \bv) + {\rm div} \left( \rho_1\phi_1 \bv\otimes \bv \right) + \nabla \breve{p}_1
    &~= 0 \label{eq: model Dumbser: mom}\\ %\phi_1 \rho_1 \mathbf{g}
  \partial_t (\rho_1\phi_1)  + {\rm div}( \rho_1\phi_1  \bv) &=~0,\label{eq: model Dumbser: mass}\\
  \partial_t \phi_1 + \bv \cdot \nabla \phi_1 &=~0.\label{eq: model Dumbser: phi}
  \end{align}
\end{subequations}
The unknowns are the specific density of phase 1 $\rho_1$, the volume fraction of phase 1 $\phi_1$, the velocity $\mathbf{v}$, and the pressure of phase 1 $p_1$. This pressure is given by the equation of state of the form $\breve{p}_1 = \breve{p}_1(\rho_1)$ for which a Tait equation of state is employed. We compare this model to the binary version of the first-order system \eqref{eq: model 1st order}, i.e.:
\begin{subequations}\label{eq: model 1st order 2 components}
  \begin{align}
   \partial_t (\rho \bv) + {\rm div} \left( \rho \bv\otimes \bv \right) + \nabla p &=~ 0, \label{eq: model  simplified 2 components: mom}\\
  \partial_t \trho_1  + {\rm div}(\trho_1  \bv) &=~0.\label{eq: eq: model simplified 2 components: mass 1}\\
  \partial_t \trho_2   + {\rm div}(\trho_2  \bv) &=~0,\label{eq: eq: model simplified 2 components: mass 2} 
  \end{align}
\end{subequations}
with unknowns $\trho_1, \trho_2, \bv$ and $p = p(\trho_1,\trho_2)$. We observe the following similarities and differences between the models:
\begin{description}
\item[$\bullet$ Structure equations]: The mass balance equation of phase 1 is present in both models: \eqref{eq: model Dumbser: mass} and \eqref{eq: eq: model simplified 2 components: mass 1}. In contrast, the momentum equations differ in the sense that \eqref{eq: model Dumbser: mom} represents the momentum balance of phase 1, whereas \eqref{eq: model  simplified 2 components: mom} describes the momentum balance of the mixture. We remark that \eqref{eq: model Dumbser} contains a single (specific) density ($\rho_1$), for which one might argue that this quantity represents the mixture density (and may denoted as $\rho$). Finally, we note that 
\eqref{eq: eq: model simplified 2 components: mass 1} may be written as:
\begin{align}
     \partial_t \phi_1  + \mathbf{v}\cdot\nabla \phi_1 = - \frac{\phi_1}{\rho_1}(\partial_t \rho_1  + {\rm div}(\rho_1  \bv)),
\end{align}
where the right-hand side is non-zero in general, showing that it is incompatible with \eqref{eq: model Dumbser: phi}.
\item[$\bullet$ Symmetry]: Model \eqref{eq: model Dumbser} is not symmetric with respect to the phases in the sense that it contains mass and momentum balance equations with respect to phase 1, but not of phase 2. Of course, the evolution equation \eqref{eq: model Dumbser: phi} is symmetric due to the saturation constraint $\phi_1+\phi_2=1$. 
\item[$\bullet$ Conservative structure]: In contrast to \eqref{eq: model 1st order 2 components}, the model \eqref{eq: model Dumbser: phi} is of non-conservative type due to the interface evolution equation \eqref{eq: model Dumbser: phi}.
\item[$\bullet$ Energy law]: Model \eqref{eq: model 1st order 2 components} satisfies the energy law \eqref{eq: energy dissipation} with $\mathscr{D}=0$. No energy law is provided for \eqref{eq: model Dumbser}.
\item[$\bullet$ Equation of state]: The equations of state of the models, $\breve{p}_1 = \breve{p}_1(\rho_1)$ for model \eqref{eq: model Dumbser} and $p = p(\trho_1,\trho_2)$ for model \eqref{eq: model 1st order 2 components} differ in two key aspects:
\begin{enumerate}
    \item their dependency on a single constituent density quantity versus dependency on each constituent density quantity
    \item the particular form of this/these density quantity/quantities, i.e. \textit{specific} constituent density versus \textit{partial} constituent densities.
\end{enumerate}%\newpage
\item[$\bullet$ Speed of sound]: Similarly as the equations of state, the constituent speed of sound of the models, $a_1^2 = {\rm d}_{\rho_1}\breve{p}_1$ for model \eqref{eq: model Dumbser} and $a_{\mA\mB}^2 = \partial_{\trho_\mB}p_{\mA}$ for model \eqref{eq: model 1st order 2 components} differ in two key aspects:
\begin{enumerate}
    \item their dependency on a single constituent density quantity versus dependency on each constituent density quantity
    \item the particular form of this/these density quantity/quantities, i.e. \textit{specific} constituent density versus \textit{partial} constituent densities.
\end{enumerate}
\item[$\bullet$ Eigenvalues]: The eigenvalues of \eqref{eq: model Dumbser} coincide with those of the first-order system \eqref{eq: model 1st order}.
    \item[$\bullet$ Hyperbolicity]: Due to the adoption of the Tait equation of state in \eqref{eq: model Dumbser}, this model is hyperbolic when $\rho_1>0$ and
$\phi_1>0$, whereas the system \eqref{eq: model 1st order} is hyperbolic when $a^2_{\mA\mB}>0$, for all $\mA,\mB=1,2$.
    \item[$\bullet$ Riemann invariants]: For model \eqref{eq: model 1st order} these are given in \eqref{eq: Riemann invariants}, while for the model \eqref{eq: model Dumbser} these are not provided.
    \item[$\bullet$ Rankine-Hugoniot conditions]: For model \eqref{eq: model 1st order} these are given in \eqref{eq: Rankine-Hugoniot conditions}, due to the non-conservative nature of the model \eqref{eq: model Dumbser} these are not well-defined.
\end{description}\vspace{0.5cm}
%\newpage
%\underline{\cite{daude2023hyperbolic}}. 
Next, we study the similarity of the model \eqref{eq: model 1st order} with that of \cite{daude2023hyperbolic}, which is given by:
\begin{subequations}\label{eq: model Daude}
  \begin{align}
   \partial_t (\rho \bv) + {\rm div} \left( \rho \bv\otimes \bv \right) + \nabla \breve{p}
    &~=0 \label{eq: model Daude: mom}\\
  \partial_t (\rho_\mA \phi_\mA )  + {\rm div}(\rho_\mA \phi_\mA  \bv) &=~0,\quad \text{ for } \mA = 1,2,3.\label{eq: model Daude: mass}
  \end{align}
\end{subequations}
where the unknowns are $\bv, \phi_\mA, \rho_\mA, \mA=1,2,3$ and $p$. These quantities are connected through the equations of state $\breve{p}_\mA = \breve{p}_\mA(\rho_\mA)$ in which the constituent pressures $\breve{p}_\mA$ are assumed individually equal and equal to the mixture pressure: $\breve{p}_1=\breve{p}_2=\breve{p}_3=\breve{p}$. We compare this model to the ternary version of the first-order system \eqref{eq: model 1st order}, i.e.:
\begin{subequations}\label{eq: model 1st order 3 components}
  \begin{align}
   \partial_t (\rho \bv) + {\rm div} \left( \rho \bv\otimes \bv \right) + \nabla p &=~ 0, \label{eq: model  simplified 3 components: mom}\\
  \partial_t \trho_\mA  + {\rm div}(\trho_\mA  \bv) &=~0, \quad \text{for} \quad \mA = 1,2,3 \label{eq: eq: model simplified 3 components: mass}
  \end{align}
\end{subequations}
with unknowns $\trho_\mA, \mA = 1,2,3, \bv$ and $p$ and the equation of state $p = p(\left\{\trho_\mA\right\}_{\mA=1,2,3})$. We observe the following similarities and differences between the models:
\begin{description}
\item[$\bullet$ Structure equations]: Through the identity $\trho_\mA=\rho_\mA\phi_\mA$ we observe that the form of the equations of \eqref{eq: model Daude} and \eqref{eq: model 1st order 3 components} is identical.
\item[$\bullet$ Symmetry]: Both model \eqref{eq: model Daude} and model \eqref{eq: model 1st order 3 components} are symmetrical with respect to the phases.
\item[$\bullet$ Conservative structure]: Both model \eqref{eq: model Daude} and model \eqref{eq: model 1st order 3 components} are conservative.
\item[$\bullet$ Energy law]: For smooth solutions, model \eqref{eq: model Daude} satisfies the energy law:
\begin{align}
    \partial_t (\rho u) + {\rm div}(\rho u \bv + \breve{p}\bv) = 0,
\end{align}
where $\rho u = \sum \trho_\mA u_\mA$, $u_\mA = \epsilon_\mA + |\bv|^2/2$ and ${\rm d}\epsilon_\mA/{\rm d}\rho_\mA = p_\mA/\rho_\mA^2= \breve{p}/\rho_\mA^2$ for some constituent energy quantity $\epsilon_\mA$.
On the other hand, \eqref{eq: model 1st order 2 components} satisfies the energy law \eqref{eq: energy dissipation} with $\mathscr{D}=0$. This may be written in the local form:
\begin{align}\label{eq: local form energy law}
    \partial_t (\rho \psi) + {\rm div}(\rho \psi \bv + p\bv) = 0.
\end{align}
\item[$\bullet$ Equation of state]: We observe the following similarities and differences regarding pressure quantities:
\begin{enumerate}
    \item Both models contain constituent pressure quantities and a single mixture pressure quantity.
    \item The constituent pressure quantities of the models, $\breve{p}_\mA=\breve{p}_\mA(\rho_\mA)$ for model \eqref{eq: model Daude} and $p_\mA(\left\{\trho_\mB\right\}_{\mB=1,2,3})$ for model \eqref{eq: model 1st order 3 components} differ in two aspects:
    \begin{enumerate}
        \item their dependency on a single constituent density quantity versus dependency on each constituent density quantity, 
        \item the particular form of this/these density quantity/quantities, i.e. \textit{specific} constituent density versus \textit{partial} constituent densities.
    \end{enumerate}
    \item The mixture pressure $\breve{p}$ of \eqref{eq: model Daude} assumes a mechanical equilibrium:
\begin{align}
    \breve{p} = \breve{p}_1=\breve{p}_2=\breve{p}_3,
\end{align}
    whereas the mixture pressure $p$ of model \eqref{eq: model 1st order 3 components} is given by Dalton's law:
    \begin{align}
        p = p_1+p_2+p_3.
    \end{align}
\end{enumerate}
\item[$\bullet$ Speed of sound]: We observe the following similarities and differences regarding speed of sound quantities:
\begin{enumerate}
    \item Both models contains both constituent speed of sound quantities and a single mixture speed of sound quantity.
    \item The constituent speed of sound quantities of the models, $\hat{a}_\mA^2 = {\rm d}_{\rho_\mA}\breve{p}_\mA$ for model \eqref{eq: model Daude} and $a_{\mA\mB}^2 = \partial_{\trho_\mB}p_{\mA}$ for model \eqref{eq: model 1st order 3 components} differ in two aspects:
    \begin{enumerate}
        \item their dependency on a single constituent density quantity versus dependency on each constituent density quantity, 
        \item the particular form of this/these density quantity/quantities, i.e. \textit{specific} constituent density versus \textit{partial} constituent densities.
    \end{enumerate}
    \item The mixture speed of sound $\hat{a}$ of \eqref{eq: model Daude} satisfies the Wood formula:
\begin{align}\label{eq: Wood formula}
    \frac{1}{\rho \hat{a}^2} = \sum_{\mA=1,2,3} \frac{\phi_\mA}{\rho_\mA \hat{a}_\mA^2}, 
\end{align}
    whereas the mixture speed of sound $a$ of model \eqref{eq: model 1st order 3 components} is given by
    \begin{align}
        \rho a^2 = \sum_{\mA,\mB=1,2,3} \trho_\mB a_{\mA\mB}^2.
    \end{align}
\end{enumerate}
\item[$\bullet$ Eigenvalues]: The eigenvalues of \eqref{eq: model Daude} and of \eqref{eq: model 1st order 3 components} and of the same form:
\begin{align}
    \lambda_1 = v-\bar{a}, \lambda_{2,3} = v, \lambda_4 = v + \bar{a},
\end{align}
with $\bar{a} = \hat{a}$ for \eqref{eq: model Daude} and $\bar{a}=a$ for \eqref{eq: model 1st order 3 components}.
    \item[$\bullet$ Hyperbolicity]: The model \eqref{eq: model Daude} is hyperbolic when $\hat{a}^2_{\mA}>0$, $\mA=1,2,3$, whereas the system \eqref{eq: model 1st order 3 components} is hyperbolic when $a^2_{\mA\mB}>0$, for all $\mA,\mB=1,2,3$.
    \item[$\bullet$ Riemann invariants]: The Riemann invariants associated with the various waves of the models \eqref{eq: model Daude} and \eqref{eq: model 1st order 3 components} are of the same form: 
\begin{subequations}\label{eq: Riemann invariants Daude}
  \begin{align}
     \lambda_1: & \quad \left\{Y_1,Y_2, v+\int_p \frac{1}{\rho \bar{a}} {\rm d}p \right\}\\
     \lambda_2,\lambda_3: & \quad \left\{v,p\right\}\\
     \lambda_{4}: &  \quad \left\{Y_1,Y_2, v-\int_p \frac{1}{\rho \bar{a}} {\rm d}p \right\},
  \end{align}
\end{subequations}
with $\bar{a} = \hat{a}$ for \eqref{eq: model Daude} and $\bar{a}=a$ for \eqref{eq: model 1st order 3 components}.
    \item[$\bullet$ Rankine-Hugoniot conditions]: The Rankine-Hugoniot conditions of both the model \eqref{eq: model Daude} and \eqref{eq: model 1st order 3 components} are well-defined and match: \begin{subequations}\label{eq: Rankine-Hugoniot conditions 2}
    \begin{align}
        \Delta \left(\trho_\mA (v-\sigma)\right) =&~ 0, \quad\mA = 1,2,3\\
        \Delta \left(\rho v(v-\sigma) + \bar{p}\right) =&~ 0,
    \end{align}
\end{subequations}
with $\bar{p}=\breve{p}$ for model \eqref{eq: model Daude} and $\bar{p}=p$ for model \eqref{eq: model 1st order 3 components}, where we recall the jump $\Delta \phi = \phi_R - \phi_L$ and the speed of the shock-wave $\sigma$. 
\end{description}\vspace{0.5cm}

\begin{remark}[Numerical difficulties Wood formula]
    The Wood formula \eqref{eq: Wood formula} has a non-monotonic variation with volume fraction $\phi_\mA$. In \cite{saurel2009simple} it is argued that this causes numerical difficulties due to inaccuracies in wave transmission across interfaces.
\end{remark}

We provide an overview of the comparison in \cref{table: overview comparison 1st order}.
\begin{center}
\begin{table}
\centering
{\small
\setlength{\extrarowheight}{2pt} % Adjust row height
                      \begin{tabular}{c >{\centering\arraybackslash}m{2.5cm} >{\centering\arraybackslash}m{3.5cm} >{\centering\arraybackslash}m{3.5cm}}
                      & \rotatebox{-10}{\textbf{Proposed model}} 
                      & \rotatebox{-10}{Dumbser \cite{dumbser2011simple}} 
                      & \rotatebox{-10}{Daude et al. \cite{daude2023hyperbolic}} \\[6pt] \thickhline\\[-4pt]
Number of phases & N  & 2 & 3    \\[6pt] \hline
Symmetric w.r.t. phases & {\color{darkgreen}\cmark}  & {\color{red}\xmark} & {\color{darkgreen}\cmark}    \\[6pt] \hline
Conservative model & {\color{darkgreen}\cmark}   & {\color{red}\xmark} &{\color{darkgreen}\cmark}    \\[6pt] \hline
Energy law & {\color{darkgreen}\cmark}   & {\color{red}\xmark} & {\color{darkgreen}\cmark}  \\[6pt] \hline
Mixture density & $\rho = \sum_\mA \trho_{\mA}$ & not defined & $\rho = \sum_\mA \trho_{\mA}$  \\[6pt] \hline
Phase pressure & $p_{\mA} = p_{\mA}(\left\{\trho_\mB\right\})$ & $p_1 = p_1(\rho_1)$ & $p_\mA = p_\mA(\rho_\mA)$ \\[6pt] \hline
Mixture pressure & $p = \sum_\mA p_{\mA}$ & not defined & $\breve{p} = \breve{p}_1=\breve{p}_2=\breve{p}_3$  \\[6pt] \hline
Phase speed of sound & $a_{\mA\mB}^2 = \partial_{\trho_\mB} p_\mA$  & $a_1^2 = \partial_{\rho_1} p_1$ & $a_{\mA}^2 = \partial_{\rho_\mA} p_\mA$ \\[6pt] \hline
Mixture speed of sound & $\rho a^2 = \sum_{\mA\mB}\trho_\mB a_{\mA\mB}^2$  & not defined & $(\rho a^2)^{-1} = \sum_{\mA}\phi_\mA (\rho_\mA a_{\mA}^2)^{-1}$ \\[6pt] \hline
Eigenvalues & $v-a, v, ..., v,v+a$  & $v-\hat{a}, v, v,v+\hat{a}$  & $v-\hat{a}, v,v,v+\hat{a}$ \\[6pt] \hline
Riemann invariants & see \eqref{eq: Riemann invariants} & not given & see \eqref{eq: Riemann invariants Daude} \\[6pt] \hline
\# mass balance laws & $N$  & $1$ & $3$ \\[6pt] \hline
\# momentum balance laws & $1$  & $1$ & $1$ \\[6pt] \hline
\# VOF type equations & $0$  & $1$ & $0$\\[6pt] \hline
Underlying theory/model  & Continuum mixture theory & Baer-Nunziato model & not given  \\[6pt] \hline 
\end{tabular}}
\caption{Comparison first-order models.}
\label{table: overview comparison 1st order}
\end{table}
\end{center}%\newpage

\subsection{Phase-field models for binary fluids}\label{subsec: Binary flows}
We show resemblance of model \eqref{eq: model binary} with two existing binary fluid models in the literature. By using \cref{lem: binary fluids Korteweg identity vol}, the model \eqref{eq: model simplified equivalent} converts into:
\begin{subequations}\label{eq: model Gomez}
  \begin{align}
  \partial_t (\rho \bv) + {\rm div} \left( \rho \bv\otimes \bv \right) + \nabla \check{p} + {\rm div} \left( \nabla Y\otimes \dfrac{\partial \check{ \Psi}}{\partial \nabla Y}  +\nabla \rho\otimes \dfrac{\partial \check{ \Psi}}{\partial \nabla \rho}  \right)&\nn\\  
    - {\rm div} \left(   \nu (2\mathbf{D}+\lambda({\rm div}\bv) \mathbf{I}) \right)-\rho\mathbf{g} &=~ 0, \label{eq: model Gomez: mom}\\
  \partial_t \rho + {\rm div}(\rho \bv) &=~ 0,\label{eq: model Gomez: rho}\\
 \partial_t (\rho Y) + {\rm div}(\rho Y \bv) - {\rm div} \left(\hat{M}\nabla\frac{\hat{\mu}_Y}{\rho}\right)  + \hat{m} \frac{\hat{\mu}_Y}{\rho} &=~0, \label{eq: model Gomez: rho c}
 \end{align}
\end{subequations}
with $p = \rho \hat{\mu}_\rho - \hat{\Psi}$. By setting $\hat{m}=0$ and inserting arbitrary source terms on the right-hand side of \eqref{eq: model Gomez: rho} and \eqref{eq: model Gomez: rho c} the model coincides with that of \cite{mukherjee2024mixtures}. 

Next, through the identification $\hat{\Psi}=\rho \hat{\psi}$, the model \eqref{eq: model simplified equivalent 2} converts into:
\begin{subequations}\label{eq: model Lowengrub 0}
  \begin{align}
  \partial_t (\rho \bv) + {\rm div} \left( \rho \bv\otimes \bv \right) + \nabla \hat{p} + {\rm div} \left(\rho \nabla Y\otimes \dfrac{\partial \hat{ \psi}}{\partial \nabla Y}  +\rho\nabla \rho\otimes \dfrac{\partial \hat{ \psi}}{\partial \nabla \rho}  \right)&\nn\\
    - {\rm div} \left(   \nu (2\mathbf{D}+\lambda({\rm div}\bv) \mathbf{I}) \right)-\rho\mathbf{g} &=~ 0, \label{eq: model  Lowengrub 0: mom}\\
  \partial_t \rho + {\rm div}(\rho \bv) &=~ 0,\label{eq: model  Lowengrub 0: rho}\\
 \partial_t (\rho Y) + {\rm div}(\rho Y \bv) - {\rm div} \left(\hat{M}\nabla\hat{\nu}_Y\right)  + \hat{m} \hat{\nu}_Y &=~0, \label{eq: model Lowengrub 0: rho c}
 \end{align}
\end{subequations}
with $\hat{p} = \rho \hat{\mu}_\rho - \hat{\Psi}$. By setting $\hat{m}=0$ and taking $\psi = \psi_0(\rho,Y) + \epsilon\|\nabla Y\|^2/2$, the model converts into the model reduces to the compressible Navier-Stokes Cahn-Hilliard model proposed in \cite{lowengrub1998quasi}:
\begin{subequations}\label{eq: model Lowengrub 1}
  \begin{align}
  \partial_t (\rho \bv) + {\rm div} \left( \rho \bv\otimes \bv \right) + \nabla \breve{p} + {\rm div} \left(\rho \nabla Y\otimes \nabla Y  \right)&\nn\\
    - {\rm div} \left(   \nu (2\mathbf{D}+\lambda({\rm div}\bv) \mathbf{I}) \right)-\rho\mathbf{g} &=~ 0, \label{eq: model  Lowengrub 1: mom}\\
  \partial_t \rho + {\rm div}(\rho \bv) &=~ 0,\label{eq: model  Lowengrub 1: rho}\\
 \partial_t (\rho Y) + {\rm div}(\rho Y \bv) - {\rm div} \left(\hat{M}\nabla \left( \dfrac{\partial \psi_0}{\partial Y} - \frac{\epsilon}{\rho} {\rm div}(\rho \nabla Y)\right)\right) &=~0, \label{eq: model Lowengrub 1: rho c}
 \end{align}
\end{subequations}
with $\breve{p} = \rho^2 \partial \psi_0/\partial \rho$. We emphasize here that the dependence of $\doublecheck{\psi}$ on $\nabla \rho$ is not taken into account in \cite{lowengrub1998quasi}. In other words, in terms of gradient dependency, \cite{lowengrub1998quasi} $\doublecheck{\psi}$ depends on $\nabla Y = \nabla ((\trho_1-\trho_2)/(\trho_1+\trho_2))$, whereas in \eqref{eq: model simplified equivalent 2} the free energy depends on $\nabla \trho_1$ and $\nabla \trho_2$.%\newpage

\begin{remark}[Matching first-order systems]\label[remark]{rmk: 1st order systems}
  The corresponding first-order systems of  \cite{mukherjee2024mixtures} and \cite{lowengrub1998quasi} have not been analyzed in those articles. However, differences between models \eqref{eq: model binary} and those models vanish for first-order systems, meaning that the properties in \cref{sec: hyperbolicity} directly carry over.
\end{remark}

We sketch an overview of the comparison of the binary models in \cref{table: overview comparison binary}.

\begin{center}
\begin{table}
\centering
{\small
\setlength{\extrarowheight}{2pt} % Adjust row height
\begin{tabular}{c >{\centering\arraybackslash}m{2.5cm} >{\centering\arraybackslash}m{3.5cm} >{\centering\arraybackslash}m{3.5cm}}
                      & \rotatebox{-10}{\textbf{Proposed model}} 
                      & \rotatebox{-10}{Lowengrub \& Truskinovsky\cite{lowengrub1998quasi}} 
                      & \rotatebox{-10}{Mukherjee \& Gomez \cite{mukherjee2024mixtures}} \\[6pt] \thickhline\\[-4pt]

Number of phases & $N$  & 2 & 2    \\[6pt] \hline
\# mass balance laws & $N$  & $2$ & $2$ \\[6pt] \hline
\# momentum balance laws & $1$  & $1$ & $1$ \\[6pt] \hline
Underlying theory/model  & Continuum mixture theory & not given  & not given  \\[6pt] \hline
Symmetric w.r.t. phases & {\color{darkgreen}\cmark} & {\color{darkgreen}\cmark} & {\color{darkgreen}\cmark}\\[6pt] \hline
Mass transfer & {\color{darkgreen}\cmark}  & {\color{red}\xmark} & {\color{red}\xmark}    \\[6pt] \hline
Energy law & {\color{darkgreen}\cmark}   & {\color{darkgreen}\cmark}&{\color{darkgreen}\cmark}/{\color{red}\xmark}$^*$    \\[6pt] \hline
Constituent quantities & see \cref{subsec: decomp free energy} & not defined & not defined \\[6pt] \hline
First-order systems$^{**}$ & see \cref{sec: hyperbolicity} & see \cref{sec: hyperbolicity} & see \cref{sec: hyperbolicity} \\[6pt] \hline
\end{tabular}}
\caption{Comparison binary models. $^*$The model of \cite{mukherjee2024mixtures} is equipped with an energy-law when the source terms in the model are set to zero. $^{**}$The first-order systems match, see \cref{rmk: 1st order systems}.}
\label{table: overview comparison binary}
\end{table}
\end{center}

%\newpage
\subsection{Incompressible $N$-phase flows}\label{subsec: Incompressible N-phase flows}
We compare the model \eqref{eq: model simplified} with the incompressible $N$-phase model given in \cite{eikelder2024unified}, which reads for phases (constituents) $\mA = 1, ..., N$:
\begin{subequations}\label{eq: intro mass}
  \begin{align}
   \partial_t (\rho \bv) + {\rm div} \left( \rho \bv\otimes \bv \right) + \sum_\mB \trho_\mB \nabla g_\mB
   % &\nn\\
    - {\rm div} \left(   \nu (2 \nabla^s \bv+\bar{\lambda}({\rm div}\bv) \mathbf{I}) \right)-\rho\mathbf{b} &=~ 0, \label{eq: intro mass: mom}\\
  \partial_t \trho_\mA  + {\rm div}(\trho_\mA  \bv) +{\rm div} \breve{\bH}_\mA  - \breve{\zeta}_\mA&=~0,\label{eq: intro mass: mass}\\
   {\rm div}\bv + \displaystyle\sum_{\mB}  \rho_\mB^{-1} \nabla \cdot \breve{\bH}_\mB  - \displaystyle\sum_{\mB}\rho_\mB^{-1}\breve{\zeta}_\mB &~= 0,\label{eq: divergence free constraint}\\
  \breve{\bH}_\mA + \sum_\mB \breve{\mathbf{M}}_{\mA\mB}\nabla g_\mB&=~0, \label{eq: intro mass: bH}\\
 % \hat{\bj}_\mA - \sum_\mB \mathbf{K}_{\mA\mB}\nabla g_\mB&=~0, \label{eq: intro mass: bj}\\
  \breve{\zeta}_\mA + \sum_\mB \breve{m}_{\mA\mB}g_\mB&=~0, \label{eq: intro mass: zeta}
  \end{align}
\end{subequations}
with the chemical potential quantities:
\begin{subequations}
    \begin{align}
                     g_\mA =&~ \dfrac{\hat{\mu}_\mA^\phi + \lambda}{\rho_\mA},\\
        \hat{\mu}_\mA^\phi =&~ \dfrac{ \partial \hat{\Psi}}{\partial \phi_\mA} - {\rm div}\dfrac{\partial \hat{\Psi}}{\partial\nabla \phi_\mA}.
\end{align}
\end{subequations}
The model \eqref{eq: intro mass} is incompressible in the sense that  $\rho_\mA$ is constant is time and space: $\rho_\mA(\mathbf{x},t) = \rho_\mA$. As such, directly enforcing the saturation condition \eqref{eq: sum phi} reveals that $\left\{\trho_\mA(\mathbf{x},t)=\rho_\mA\phi_\mA(\mathbf{x},t)\right\}_{\mA=1,...,N}$ constitute $N-1$ independent variables. Instead, formulation \eqref{eq: intro mass} enforces the saturation constraint through \eqref{eq: divergence free constraint}, therefore avoiding ambiguities. The set of variables is completed, analogously to a single-phase system, with Lagrange multiplier $\lambda$ that enforces \eqref{eq: divergence free constraint}. 

Obviously, since the models \eqref{eq: model full} and \eqref{eq: model simplified} have largely been derived in a similar fashion, both share key features: $N$ mass balance laws, $1$ momentum balance law, and similar closure models for diffusive fluxes and mass transfer terms. The key difference lies in the presence of \eqref{eq: divergence free constraint}. This manifests itself in the form of the chemical potentials, namely $g_\mA = \hat{\mu}_\mA + \rho_\mA^{-1}\lambda$. The model may be written into the form
\begin{subequations}\label{eq: model full incomp}
  \begin{align}
   \partial_t (\rho \bv) + {\rm div} \left( \rho \bv\otimes \bv \right) + \nabla p + {\rm div} \left(\sum_\mA \nabla \trho_\mA \otimes \dfrac{\partial \hat{\Psi}}{\partial \nabla \trho_\mA} \right)& \nn\\
    - {\rm div} \left(   \nu (2\mathbf{D}+\lambda({\rm div}\bv) \mathbf{I}) \right)-\rho\mathbf{g} &=~ 0, \label{eq: model full: mom incomp}\\
   \partial_t \trho_\mA  + {\rm div}(\trho_\mA  \bv) +{\rm div} \breve{\bH}_\mA  - \breve{\zeta}_\mA &=~0,\\
   {\rm div}\bv + \displaystyle\sum_{\mB}  \rho_\mB^{-1} \nabla \cdot \breve{\bH}_\mB  - \displaystyle\sum_{\mB}\rho_\mB^{-1}\breve{\zeta}_\mB &~= 0,
  \end{align}
\end{subequations}
where $p = \lambda + \sum_\mA \trho_\mA \hat{\mu}_\mA - \hat{\Psi}$, which closely resembles \eqref{eq: model full}. This shows that the pressure in \eqref{eq: model full incomp} is the superposition of the thermodynamical pressure $\sum_\mA \trho_\mA \hat{\mu}_\mA - \hat{\Psi}$ and the Lagrange multiplier $\lambda$.

The observed strong connection is in line with expectations (just as for single-phase flow), reflecting the natural link between compressible and incompressible models. However, this does not hold for existing compressible two-phase flow models, where establishing their incompressible counterparts is not straightforward.

\section{Conclusion}\label{sec: discussion} 
In this work, we have developed a thermodynamical framework for compressible, isothermal $N$-phase mixtures. The framework is formulated within the principles of continuum mixture theory, which serves as a fundamental basis for modeling multi-physics systems. We have established some connections with existing phase-field and compressible two-phase flow models, in an effort to initiate to the gap between these fields of research. Specifically, the proposed formulation extends the Navier-Stokes-Korteweg (NSK) model to multiple fluids and adapts the Navier-Stokes-Cahn-Hilliard/Allen-Cahn (NSCH/AC) models to general $N$-phase compressible mixtures. A key distinction from existing compressible two-phase flow models is that reducing the complexity of those models (i.e., using fewer equations) often leads to the loss or degradation of essential thermodynamic properties. In contrast, the framework proposed in this work maintains thermodynamic consistency regardless of its level of complexity.
%\newpage
Several open challenges remain for future research. We highlight some key directions. First, the framework offers a foundation for studying spinodal decomposition in compressible mixtures, with potential applications in phase-change problems such as evaporation and condensation. Second, by incorporating free energies with explicit interface width parameters, one can investigate sharp-interface limits, both formally and rigorously. Other important extensions include the incorporation of chemical reactions and non-isothermal effects. Furthermore, an essential avenue for future work, beyond the scope of this paper, is the development of efficient numerical schemes that accurately capture interfacial dynamics while preserving key thermodynamic properties.

\appendix

\section{Identities Korteweg tensors}\label{sec: appendix: ID Korteweg}

\begin{lemma}[Identity Korteweg stresses]\label[lemma]{lem: korteweg id}
We have the following identity for the pressure and Korteweg stresses:
  \begin{align}
      \nabla p + {\rm div} \left( \sum_\mA \nabla \trho_\mA \otimes \dfrac{\partial \hat{\Psi}}{\partial \nabla \trho_\mA} \right) = \sum_\mA \trho_\mA\nabla\hat{\mu}_\mA,
  \end{align}
where we recall $p=\sum_\mA\hat{\mu}_\mA\trho_\mA-\hat{\Psi}$.
\end{lemma}
\begin{proof}
This follows from the sequence of identities:
  \begin{align}
      &\nabla p + {\rm div} \left( \sum_\mA \nabla \trho_\mA \otimes \dfrac{\partial \hat{\Psi}}{\partial \nabla \trho_\mA} \right)=\nn\\
      &\nabla \left(\sum_\mA\hat{\mu}_\mA\trho_\mA-\hat{\Psi}\right) + \sum_\mA \left(\nabla \trho_\mA {\rm div} \left(\dfrac{\partial \hat{\Psi}}{\partial \nabla \trho_\mA}\right) + (\mathbf{H} \trho_\mA)\dfrac{\partial \hat{\Psi}}{\partial \nabla \trho_\mA}\right)=\nn\\
      &\sum_\mA \left(\nabla \left(\hat{\mu}_\mA\trho_\mA \right) - \dfrac{\partial \hat{\Psi}}{\partial  \trho_\mA}\nabla \trho_\mA - (\mathbf{H} \trho_\mA) \dfrac{\partial \hat{\Psi}}{\partial \nabla \trho_\mA}+  \nabla \trho_\mA {\rm div} \left(\dfrac{\partial \hat{\Psi}}{\partial \nabla \trho_\mA}\right) + (\mathbf{H} \trho_\mA)\dfrac{\partial \hat{\Psi}}{\partial \nabla \trho_\mA}\right)=\nn\\
      %&\nabla (\hat{\mu}_1\trho_1+\hat{\mu}_2\trho_2 )-\hat{\mu}_1\nabla \trho_1-\hat{\mu}_2\nabla \trho_2=\nn\\
      &\sum_\mA \trho_\mA\nabla\hat{\mu}_\mA,
  \end{align}
where $\mathbf{H}\trho_\mA$ denotes the Hessian of the partial density $\trho_\mA$.
\end{proof}

\begin{lemma}[Identity Korteweg stresses -- Binary fluids -- Volume-measure]\label[lemma]{lem: binary fluids Korteweg identity vol}
The following identity holds for the pressure and Korteweg stresses in binary fluids:
  \begin{align}
      \nabla \check{p} + {\rm div} \left( \nabla Y\otimes \dfrac{\partial \check{ \Psi}}{\partial \nabla Y}  +\nabla \rho\otimes \dfrac{\partial \check{ \Psi}}{\partial \nabla \rho}  \right)= \rho\nabla\check{\mu}_\rho -\check{\mu}_Y\nabla Y,
  \end{align}
where we recall $\check{p} = \check{\mu}_\rho \rho - \check{\Psi}$.
\end{lemma}

\begin{proof}
This follows from the sequence of identities:
\begin{align}
  &\nabla \check{p} + {\rm div} \left( \nabla Y\otimes \dfrac{\partial \check{ \Psi}}{\partial \nabla Y}  +\nabla \rho\otimes \dfrac{\partial \check{ \Psi}}{\partial \nabla \rho}  \right) =\nn\\
  &\nabla(\check{\mu}_\rho\rho) - \nabla \check{\Psi} + \nabla Y {\rm div}\dfrac{\partial \check{ \Psi}}{\partial \nabla Y} + (\mathbf{H}Y)\dfrac{\partial \check{ \Psi}}{\partial \nabla Y} + \nabla \rho {\rm div}\dfrac{\partial \check{ \Psi}}{\partial \nabla \rho} + (\mathbf{H}\rho)\dfrac{\partial \check{ \Psi}}{\partial \nabla \rho}=\nn\\
  &\nabla(\check{\mu}_\rho\rho) -\nabla Y\dfrac{\partial \check{ \Psi}}{\partial Y}-\nabla \rho\dfrac{\partial \check{ \Psi}}{\partial \rho} + \nabla Y {\rm div}\dfrac{\partial \check{ \Psi}}{\partial \nabla Y} + \nabla \rho {\rm div}\dfrac{\partial \check{ \Psi}}{\partial \nabla \rho} =\nn\\
  %&\nabla(\upsilon\rho) -\mu\nabla Y-\upsilon\nabla \rho =\nn\\
  &\rho\nabla\check{\mu}_\rho -\check{\mu}_Y\nabla Y,
\end{align}
with $\mathbf{H}\rho$ and $\mathbf{H}Y$ the Hessians $\rho$ and $Y$.
\end{proof}

\begin{lemma}[Identity Korteweg stresses -- Binary fluids -- Mass measure]\label[lemma]{lem: binary fluids Korteweg identity mass}
The following identity holds for the pressure and Korteweg stresses in binary fluids:
  \begin{align}
      \nabla \doublecheck{p} + {\rm div} \left( \nabla Y\otimes \rho \dfrac{\partial \doublecheck{\psi}}{\partial \nabla Y}  +\nabla \rho\otimes \rho \dfrac{\partial \check{\psi}}{\partial \nabla \rho}  \right)= \rho \nabla (\rho \doublecheck{\mu}_\rho + \doublecheck{\psi})-\rho \doublecheck{\mu}_Y\nabla Y,
  \end{align}
where we recall $\doublecheck{p} = \rho^2 \doublecheck{\mu}_\rho$.
\end{lemma}

\begin{proof}
This follows from the sequence of identities:
\begin{align}
  &\nabla \doublecheck{p} + {\rm div} \left( \nabla Y\otimes \rho \dfrac{\partial \doublecheck{\psi}}{\partial \nabla Y}  +\nabla \rho\otimes \rho \dfrac{\partial \check{\psi}}{\partial \nabla \rho}  \right) =\nn\\
  &\rho^2 \nabla \doublecheck{\mu}_\rho + 2\rho  \doublecheck{\mu}_\rho\nabla\rho + \nabla Y {\rm div}\left(\rho\dfrac{\partial \check{ \psi}}{\partial \nabla Y} \right)+ (\mathbf{H}Y)\rho\dfrac{\partial \check{ \psi}}{\partial \nabla Y} + \nabla \rho {\rm div}\left(\rho \dfrac{\partial \doublecheck{ \psi}}{\partial \nabla \rho} \right)+ (\mathbf{H}\rho)\rho\dfrac{\partial \doublecheck{ \psi}}{\partial \nabla \rho}=\nn\\
  &\rho^2 \nabla \doublecheck{\mu}_\rho + 2\rho  \doublecheck{\mu}_\rho\nabla\rho - \rho  \doublecheck{\mu}_\rho\nabla\rho + \rho\dfrac{\partial \doublecheck{ \psi}}{\partial \rho}  + \rho\nabla \doublecheck{\psi} -  \rho\dfrac{\partial \doublecheck{ \psi}}{\partial \rho} \nabla \rho -  \rho\dfrac{\partial \doublecheck{ \psi}}{\partial Y} \nabla Y  - \rho \doublecheck{\mu}_Y \nabla Y + \rho \dfrac{\partial \doublecheck{ \psi}}{\partial Y} \nabla Y  =\nn\\
  %&\nabla(\upsilon\rho) -\mu\nabla Y-\upsilon\nabla \rho =\nn\\
  &\rho \nabla (\rho \doublecheck{\mu}_\rho + \doublecheck{\psi})-\rho \doublecheck{\mu}_Y\nabla Y.
\end{align}
\end{proof}
\section{Variable transformation - binary mixture}
In this section we provide details on the transformation of constituent quantities to mixture quantities for a binary mixture. The variable transformations are given by:
\begin{align}
    \trho_1 = \rho \frac{1+Y}{2}, \quad \trho_2 =\rho \frac{1-Y}{2}, \quad \rho  = \trho_1+\trho_2, \quad  Y = \dfrac{\trho_1-\trho_2}{\trho_1+\trho_2}.
\end{align}
The derivatives of the scalar quantities are:
\begin{subequations}\label{eq: appendix: var trans scalar}
\begin{align}
    \dfrac{\partial \trho_1}{\partial \rho} =&~ \dfrac{1+Y}{2}, &\quad 
    \dfrac{\partial \trho_2}{\partial \rho} =&~ \dfrac{1-Y}{2}, \\
    \dfrac{\partial \trho_1}{\partial Y} =&~ \dfrac{\rho}{2}, &\quad 
    \dfrac{\partial \trho_2}{\partial Y} =&~-\dfrac{\rho}{2}, \\
    \dfrac{\partial \rho}{\partial \trho_1} =&~ 1, &\quad 
    \dfrac{\partial \rho}{\partial \trho_2} =&~ 1, \\
    \dfrac{\partial Y}{\partial \trho_1} =&~ \dfrac{2\trho_2}{(\trho_1+\trho_2)^2}, &\quad
    \dfrac{\partial Y}{\partial \trho_2} =&~ -\dfrac{2\trho_1}{(\trho_1+\trho_2)^2}.
\end{align}
\end{subequations}%\newpage
The derivatives of the gradient quantities are:
\begin{subequations}\label{eq: appendix: var trans grad}
  \begin{align}
    \dfrac{\partial \nabla \rho}{\partial \trho_1} =&~ \dfrac{\partial \nabla \rho}{\partial \trho_2} = 0, \\
    \dfrac{\partial \nabla Y}{\partial \trho_1} =& -4 \frac{\trho_2}{(\trho_1+\trho_2)^3}\nabla \trho_1 + 2 \frac{\trho_1-\trho_2}{(\trho_1+\trho_2)^3}\nabla \trho_2, \\
    \dfrac{\partial \nabla Y}{\partial \trho_2} =& ~2 \frac{\trho_1-\trho_2}{(\trho_1+\trho_2)^3}\nabla \trho_1 + 4 \frac{\trho_1}{(\trho_1+\trho_2)^3}\nabla \trho_2.
\end{align}
\end{subequations}
%\section{Transformation Korteweg tensor}
\begin{lemma}[Relation Korteweg stresses]
  The Korteweg stresses are related via:
  \begin{align}
      \nabla \trho_1 \otimes \dfrac{\partial \hat{\Psi}}{\partial \nabla \trho_1}+\nabla \trho_2 \otimes \dfrac{\partial \hat{\Psi}}{\partial \nabla \trho_2} = \nabla \rho \otimes \dfrac{\partial \check{\Psi}}{\partial \nabla \rho} + \nabla Y \otimes \dfrac{\partial \check{\Psi}}{\partial \nabla Y}.
  \end{align}
\end{lemma}
\begin{proof}
The chain rule and the variable transformations \eqref{eq: appendix: var trans scalar} provide:
\begin{subequations}
\begin{align}
  \nabla \trho_1 =&~ %\dfrac{\partial \trho_1}{\partial \rho} \nabla \rho + \dfrac{\partial \trho_1}{\partial Y} \nabla Y = 
  \frac{1+Y}{2}\nabla \rho + \frac{\rho}{2} \nabla Y,\\
  \nabla \trho_2 =&~ %\dfrac{\partial \trho_2}{\partial \rho} \nabla \rho + \dfrac{\partial \trho_2}{\partial Y} \nabla Y = 
  \frac{1-Y}{2}\nabla \rho - \frac{\rho}{2} \nabla Y.
\end{align}
\end{subequations}
Next, the chain rule and the variable transformations \eqref{eq: appendix: var trans grad} provide:
\begin{subequations}
\begin{align}
  \dfrac{\partial \hat{\Psi}}{\partial \nabla \trho_1} %&=~  \dfrac{\partial \check{\Psi}}{\partial \nabla \trho_1} \nn\\
%&=~  \dfrac{\partial \check{\Psi}}{\partial \nabla \rho} \dfrac{\partial \rho}{\partial \trho_1} + \dfrac{\partial \check{\Psi}}{\partial \nabla Y} \dfrac{\partial Y}{\partial \trho_1} \nn  \\
&=~  \dfrac{\partial \check{\Psi}}{\partial \nabla \rho}  + \dfrac{\partial \check{\Psi}}{\partial \nabla Y} \dfrac{1-Y}{\rho},\\
  \dfrac{\partial \hat{\Psi}}{\partial \nabla \trho_2} %&=~  \dfrac{\partial \check{\Psi}}{\partial \nabla \trho_2} \nn\\
%&=~  \dfrac{\partial \check{\Psi}}{\partial \nabla \rho} \dfrac{\partial \rho}{\partial \trho_2} + \dfrac{\partial \check{\Psi}}{\partial \nabla Y} \dfrac{\partial Y}{\partial \trho_2} \nn  \\
&=~  \dfrac{\partial \check{\Psi}}{\partial \nabla \rho} - \dfrac{\partial \check{\Psi}}{\partial \nabla Y} \dfrac{1+Y}{\rho}.
\end{align}
\end{subequations}
Combining the above expressions gives:
\begin{subequations}
\begin{align}
\nabla \trho_1 \otimes \dfrac{\partial \hat{\Psi}}{\partial \nabla \trho_1} =&~ \frac{1+Y}{2} \nabla \rho \otimes \dfrac{\partial \check{\Psi}}{\partial \nabla \rho} + \frac{1-Y}{2}\nabla Y \otimes \dfrac{\partial \check{\Psi}}{\partial \nabla Y}\nn\\
&~+ \frac{1-Y^2}{2\rho} \nabla \rho \otimes \dfrac{\partial \check{\Psi}}{\partial \nabla Y} + \frac{\rho}{2} \nabla Y \otimes \dfrac{\partial \check{\Psi}}{\partial \nabla \rho},\\
\nabla \trho_2 \otimes \dfrac{\partial \hat{\Psi}}{\partial \nabla \trho_2} =&~ \frac{1-Y}{2} \nabla \rho \otimes \dfrac{\partial \check{\Psi}}{\partial \nabla \rho} + \frac{1+Y}{2}\nabla Y \otimes \dfrac{\partial \check{\Psi}}{\partial \nabla Y}\nn\\
&~- \frac{1-Y^2}{2\rho} \nabla \rho \otimes \dfrac{\partial \check{\Psi}}{\partial \nabla Y} - \frac{\rho}{2} \nabla Y \otimes \dfrac{\partial \check{\Psi}}{\partial \nabla \rho}.
\end{align}
\end{subequations}
Addition of these two identities provides the result.
\end{proof}

%\section{Transformation chemical potentials}\label{sec: appendix: Transformation chemical potentials}
\begin{lemma}[Relations chemical potentials]\label[lemma]{lem: appendix: relation chem}
  The chemical potentials are related via:
  \begin{subequations}
  \begin{align}
      \hat{\mu}_1\trho_1+\hat{\mu}_2\trho_2 =&~ \hat{\mu}_\rho \rho,\\
      \trho_1 \nabla \hat{\mu}_1+\trho_2\nabla\hat{\mu}_2=&~\rho \nabla \hat{\mu}_\rho - \hat{\mu}_Y \nabla Y,\\
      \hat{\mu}_1-\hat{\mu}_2 =&~ 2\frac{\hat{\mu}_Y}{\rho}.
  \end{align}
\end{subequations}
\end{lemma}
\begin{proof}
The first and third identity follow from the chain rule:
\begin{subequations}
    \begin{align}
        \hat{\mu}_\rho =&~ \hat{\mu}_1 \frac{\partial \trho_1}{\partial \rho}+\hat{\mu}_2 \frac{\partial \trho_2}{\partial \rho} = \frac{1}{\rho}(\hat{\mu}_1 \trho_1+\hat{\mu}_2 \trho_2),\\
        \hat{\mu}_Y =&~ \hat{\mu}_1 \frac{\partial \trho_1}{\partial Y}+\hat{\mu}_2 \frac{\partial \trho_2}{\partial Y} = \frac{\rho}{2} (\hat{\mu}_1 -\hat{\mu}_2).
    \end{align}
\end{subequations}
Taking the gradient of the first identity yields:
  \begin{align}
      \hat{\mu}_1 \nabla \trho_1+\hat{\mu}_2 \nabla \trho_2 +    \trho_1\nabla \hat{\mu}_1 +\trho_2\nabla \hat{\mu}_2=&~ \hat{\mu}_\rho \nabla \rho+\rho \nabla \hat{\mu}_\rho.
  \end{align}
  Noting the identity
  \begin{align}
      \nabla \hat{\Psi} = \hat{\mu}_1 \nabla \trho_1+\hat{\mu}_2 \nabla \trho_2  = \hat{\mu}_\rho \nabla \rho + \hat{\mu}_Y \nabla Y
  \end{align}
  completes the proof.
\end{proof}
\section{Constitutive modeling binary mixture}\label{sec: appendix: alternative constitutive modeling: rho c Psi}
This section presents alternative derivations of the modeling framework for binary fluids. In \cref{subsec: volume measure free energy}, we detail the derivation using a volume-measure-based free energy ($\Psi$), while \cref{subsec: mass measure free energy} covers the derivation based on a mass-measure-based free energy ($\psi$), which are related via $\Psi = \rho \psi$. We utilize the system of balance laws \eqref{eq: BL constitutive} in mixture variables:
\begin{subequations}\label{eq: BL constitutive 2}
  \begin{align}
        \partial_t \rho  + {\rm div}(\rho  \bv)  &=~ 0, \label{eq: appendix mass 1}\\
        \partial_t (\rho Y) + {\rm div}(\rho Y \bv) + {\rm div} \bH &=~ \zeta,  \label{eq: appendix mass 2}\\
        \partial_t (\rho \bv) + {\rm div} \left( \rho\bv\otimes \bv \right) -  {\rm div} \mathbf{T} -  \rho \mathbf{b} &=~0,\\
        \mathbf{T}-\mathbf{T}^T &=~0,
  \end{align}
\end{subequations}
with $Y=Y_1-Y_2$, $\bH = \bH_1-\bH_2$ and $\zeta = \zeta_1-\zeta_2$. In both case we use the energy-dissipative design condition \eqref{eq: energy dissipation} where the energy is given in \eqref{eq: total energy}.

\subsection{Volume-measure free energy}\label{subsec: volume measure free energy}
We postulate the free energy to pertain to the constitutive class:
\begin{align}\label{eq: appendix: class Psi}
  \Psi = \check{\Psi}(\rho, \nabla \rho, Y,\nabla Y),  
\end{align}
and introduce the chemical potentials:
\begin{subequations}
\begin{align}
    \check{\mu}_Y =&~ \dfrac{ \partial \check{\Psi}}{\partial Y} - {\rm div}\dfrac{\partial \check{\Psi}}{\partial\nabla Y},\\
    \check{\mu}_\rho =&~ \dfrac{ \partial \check{\Psi}}{\partial \rho} - {\rm div}\dfrac{\partial \check{\Psi}}{\partial\nabla \rho}.
\end{align}
\end{subequations}
We proceed with the evaluation of the evolution of the energy \eqref{eq: appendix: class Psi}. Following the same steps as in \cref{sec: const mod: subsec: model restr}, we applying Reynolds transport theorem, the divergence theorem, and subsequently expand the derivatives to find:
\begin{align}\label{eq: appendix: Psi derivation 1}
    \dfrac{{\rm d}}{{\rm d}t}\displaystyle\int_{\mathcal{R}(t)} \check{\Psi} ~{\rm d}v  &~= \displaystyle\int_{\mathcal{R}(t)} \dfrac{\partial \check{\Psi}}{\partial \rho} \dot{\rho} + \dfrac{\partial \check{\Psi}}{\partial \nabla \rho}\cdot \dot{\overline{\nabla \rho}} + \dfrac{\partial \check{\Psi}}{\partial Y} \dot{Y} + \dfrac{\partial \check{\Psi}}{\partial \nabla Y}\cdot \dot{\overline{\nabla Y}} + \check{\Psi} ~{\rm div} \bv~{\rm d}v,
\end{align}
where we recall the notation
\begin{align}
    \dot{\overline{\nabla \omega}}:=\partial_t \nabla \omega + \bv\cdot \nabla \left(\nabla \omega \right)
\end{align}
for the material derivative of the quantity $\nabla \omega$.
By substituting the identity \eqref{eq: relation grad phi} for $\omega = Y, \rho$, and subsequently integrating by parts, we arrive at:
\begin{align}
    \dfrac{{\rm d}}{{\rm d}t}\displaystyle\int_{\mathcal{R}(t)} \check{\Psi} ~{\rm d}v  =&~ \displaystyle\int_{\mathcal{R}(t)} \check{\mu}_\rho \dot{\rho} + \check{\mu}_Y \dot{Y} - \left(\nabla Y \otimes \dfrac{\partial \check{\Psi}}{\partial \nabla Y} + \nabla \rho \otimes \dfrac{\partial \check{\Psi}}{\partial \nabla \rho}\right): \nabla \bv  + \check{\Psi}~{\rm div} \bv~{\rm d}v  \nn\\
    &~+ \displaystyle\int_{\partial \mathcal{R}(t)}\left(\dot{Y} \dfrac{\partial \check{\Psi}}{\partial \nabla Y}+\dot{\rho} \dfrac{\partial \check{\Psi}}{\partial \nabla \rho}\right)\cdot \boldsymbol{\nu} ~{\rm d}a.
\end{align}
Substituting the mass balance equations \eqref{eq: appendix mass 1}-\eqref{eq: appendix mass 2} and applying integration by parts leads to:
\begin{align}\label{eq: appendix: Psi 2}
    \dfrac{{\rm d}}{{\rm d}t}\displaystyle\int_{\mathcal{R}(t)} \check{\Psi} ~{\rm d}v = &~ \displaystyle\int_{\mathcal{R}(t)} -\check{\mu}_\rho \rho {\rm div}\bv+\nabla \left(\dfrac{\check{\mu}_Y}{\rho}\right) \cdot \bH  - \left(\nabla \rho \otimes \dfrac{\partial \check{\Psi}}{\partial \nabla \rho}+\nabla Y \otimes \dfrac{\partial \check{\Psi}}{\partial \nabla Y}\right): \nabla \bv   \nn\\
    &~\quad \quad + \check{\Psi}~{\rm div} \bv + \left(\dfrac{\check{\mu}_Y}{\rho}\right)\zeta~{\rm d}v\nn\\
    &~+ \displaystyle\int_{\partial \mathcal{R}(t)}\left(-\rho^{-1}\check{\mu}_Y \bH+\dot{\rho} \dfrac{\partial \check{\Psi}}{\partial \nabla \rho}+\dot{Y} \dfrac{\partial \check{\Psi}}{\partial \nabla Y}\right)\cdot \boldsymbol{\nu} ~{\rm d}a.
\end{align}
Taking the sum of \eqref{eq: appendix: Psi 2} and \eqref{eq: kin grav evo}, we obtain:
\begin{align}\label{eq: appendix: second law subst 1}
    \dfrac{{\rm d}}{{\rm d}t} \mathscr{E} = &~ \displaystyle\int_{\partial \mathcal{R}(t)}\left(\bv^T\mathbf{T}-\rho^{-1}\check{\mu}_Y \bH+\dot{Y} \dfrac{\partial \check{\Psi}}{\partial \nabla Y}+\dot{\rho} \dfrac{\partial \check{\Psi}}{\partial \nabla \rho}\right)\cdot \boldsymbol{\nu} ~{\rm d}a \nn\\
    &~- \displaystyle\int_{\mathcal{R}(t)}  \left( \mathbf{T}  +\check{\mu}_\rho \rho {\rm div}\bv-\nabla \left(\dfrac{\check{\mu}_Y}{\rho}\right)\cdot \bH  + \left(\nabla Y \otimes \dfrac{\partial \check{\Psi}}{\partial \nabla Y}+\nabla \rho \otimes \dfrac{\partial \check{\Psi}}{\partial \nabla \rho}\right): \nabla \bv  \right. \nn\\
    &~\quad \quad \left.- \check{\Psi}~{\rm div} \bv - \left(\dfrac{\check{\mu}_Y}{\rho}\right)\zeta \right)~{\rm d}v,
\end{align}
in which we identify the rate of work and the dissipation as:
\begin{subequations}\label{eq: appendix: W, D}
\begin{align}
    \mathscr{W} =&~   \displaystyle\int_{\partial \mathcal{R}(t)}\left(\bv^T\mathbf{T}-\left(\dfrac{\check{\mu}_Y}{\rho}\right) \bH+\dot{Y} \dfrac{\partial \check{\Psi}}{\partial \nabla Y}+\dot{\rho} \dfrac{\partial \check{\Psi}}{\partial \nabla \rho}\right)\cdot \boldsymbol{\nu} ~{\rm d}a,\\
    \mathscr{D} =&~  \displaystyle\int_{\mathcal{R}(t)}  \left( \mathbf{T}  + \nabla Y \otimes \dfrac{\partial \check{\Psi}}{\partial \nabla Y}+\nabla \rho \otimes \dfrac{\partial \check{\Psi}}{\partial \nabla \rho} + \left(\check{\mu}_\rho \rho-\check{\Psi}\right)\mathbf{I} \right): \nabla \bv  \nn\\
    &\quad \quad \quad -\nabla \left(\dfrac{\check{\mu}_Y}{\rho}\right) \cdot \bH   -  \left(\dfrac{\check{\mu}_Y}{\rho}\right) \zeta ~{\rm d}v.\label{eq: def diffusion 2}
\end{align}
\end{subequations}
Due to the arbitrariness of the control volume, the energy-dissipation law holds when the local inequality is satisfied:\begin{align}\label{eq: appendix: second law 4}
    \left( \mathbf{T}  + \nabla Y \otimes \dfrac{\partial \check{\Psi}}{\partial \nabla Y} + \nabla \rho \otimes \dfrac{\partial \check{\Psi}}{\partial \nabla \rho} + \left(\check{\mu}_\rho \rho-\check{\Psi}\right)\mathbf{I} \right): \nabla \bv \nn\\
    -\nabla \left(\dfrac{\check{\mu}_Y}{\rho}\right) \cdot \bH  -  \left(\dfrac{\check{\mu}_Y}{\rho}\right) \zeta \geq 0.
\end{align}%\newpage
Based on the inequality \eqref{eq: appendix: second law 4}, we now restrict ourselves to stress tensors $\mathbf{T}$, diffusive fluxes $\bh$ and mass fluxes $\zeta$ belonging to the constitutive classes:
\begin{subequations}\label{eq: appendix: class T h}
\begin{align}
    \mathbf{T} =&~ \check{\mathbf{T}}(\nabla \bv, \rho,\nabla \rho,Y, \nabla Y, \check{\mu}_\rho, \nabla \check{\mu}_\rho, \check{\mu}_Y, \nabla \check{\mu}_Y), \label{eq: class T 2}\\
    \bH =&~ \check{\bH}(\rho,\nabla \rho,Y, \nabla Y, \check{\mu}_\rho, \nabla \check{\mu}_\rho, \check{\mu}_Y, \nabla \check{\mu}_Y).\label{eq: appendix: class h}\\
    \zeta =&~ \check{\zeta}(\rho,Y, \check{\mu}_\rho, \check{\mu}_Y).\label{eq: appendix: class gamma}
\end{align}
\end{subequations}
We select closure models that render each of the terms in \eqref{eq: appendix: second law 4} separately non-negative:
\begin{subequations}\label{eq: appendix: choices}
    \begin{align}
    \check{\mathbf{T}} =&~ - \nabla Y\otimes \dfrac{\partial \check{ \Psi}}{\partial \nabla Y} - \nabla \rho \otimes \dfrac{\partial \check{ \Psi}}{\partial \nabla \rho} - \left(\check{\mu}_\rho \rho-\check{\Psi}\right)\mathbf{I}  + \nu (2\mathbf{D}+\lambda({\rm div}\bv) \mathbf{I}),\label{eq: appendix: stress tensor choice}\\
    \check{\bH} =&~ -\check{\mathbf{M}} \nabla \left(\dfrac{\check{\mu}_Y}{\rho}\right),\\
    \check{\zeta} =&~ -\check{m}\left(\dfrac{\check{\mu}_Y}{\rho}\right),
\end{align}
\end{subequations}
a positive definite mobility tensor $ \check{\mathbf{M}} $ with the same dependencies as \eqref{eq: appendix: class h}, satisfying $ \check{\mathbf{M}}|_{Y=0,1} = 0 $, and a non-negative scalar mobility $ \check{m} $ with the same dependencies as \eqref{eq: appendix: class gamma}, also satisfying $ \check{m}|_{Y=0,1} = 0 $.

This completes the constitutive modeling, and the system is given by:
\begin{subequations}\label{eq: appendix: model v with const mod}
  \begin{align}
   \partial_t (\rho \bv) + {\rm div} \left( \rho \bv\otimes \bv \right) + \nabla \check{p} + {\rm div} \left( \nabla Y\otimes \dfrac{\partial \check{ \Psi}}{\partial \nabla Y}  +\nabla \rho\otimes \dfrac{\partial \check{ \Psi}}{\partial \nabla \rho}  \right) & \nn\\
    - {\rm div} \left(   \nu (2\mathbf{D}+\lambda({\rm div}\bv) \mathbf{I}) \right)-\rho\mathbf{g} &=~ 0, \label{eq: appendix: model v with const mod: mom}\\
 \partial_t \rho + {\rm div}(\rho \bv) &=~ 0, \label{eq: appendix: model v with const mod: cont} \\
  \partial_t (\rho Y) + {\rm div}(\rho Y \bv) - {\rm div} \left(\check{\mathbf{M}}\nabla \left(\dfrac{\check{\mu}_Y}{\rho}\right)\right)  + \check{m} \left(\dfrac{\check{\mu}_Y}{\rho}\right) &=~0,\label{eq: appendix: model v with const mod: PF}\\
  \check{\mu}_Y - \dfrac{\partial \check{\Psi}}{\partial Y}+{\rm div}   \dfrac{\partial \check{\Psi}}{\partial \nabla Y} &=~0,\\
  \check{\mu}_\rho - \dfrac{ \partial \check{\Psi}}{\partial \rho} + {\rm div}\dfrac{\partial \check{\Psi}}{\partial\nabla \rho} &=~0,
  \end{align}
\end{subequations}
with the pressure $\check{p} =\check{\mu}_\rho \rho-\check{\Psi}$. Analogously to \cref{lem: identity Korteweg stresses}, the Korteweg stress and pressure contributions in the momentum equation may be expressed using the chemical potential quantities.
\begin{lemma}[Identity Korteweg stresses - Binary fluids -- Volume measure]
The following identity holds for the pressure and Korteweg stresses in binary fluids:
  \begin{align}
      \nabla \check{p} + {\rm div} \left( \nabla Y\otimes \dfrac{\partial \check{ \Psi}}{\partial \nabla Y}  +\nabla \rho\otimes \dfrac{\partial \check{ \Psi}}{\partial \nabla \rho}  \right)= \rho\nabla\check{\mu}_\rho -\check{\mu}_Y\nabla Y.
  \end{align}
\end{lemma}
\begin{proof}
See \cref{lem: binary fluids Korteweg identity vol}.
\end{proof}%\newpage
Applying the lemma converts the model into the compact form:
\begin{subequations}\label{eq: appendix: model v with const mod 2}
  \begin{align}
   \partial_t (\rho \bv) + {\rm div} \left( \rho \bv\otimes \bv \right) + \rho\nabla\check{\mu}_\rho -\check{\mu}_Y\nabla Y %& \nn\\
    - {\rm div} \left(   \nu (2\mathbf{D}+\lambda({\rm div}\bv) \mathbf{I}) \right)-\rho\mathbf{g} &=~ 0, \label{eq: appendix: model v with const mod: mom 2}\\
 \partial_t \rho + {\rm div}(\rho \bv) &=~ 0, \label{eq: appendix: model v with const mod: cont 2} \\
  \partial_t (\rho Y) + {\rm div}(\rho Y \bv) - {\rm div} \left(\check{\mathbf{M}}\nabla \left(\dfrac{\check{\mu}_Y}{\rho}\right)\right)  + \check{m} \left(\dfrac{\check{\mu}_Y}{\rho}\right) &=~0,\label{eq: appendix: model v with const mod: PF 2}\\
  \check{\mu}_Y - \dfrac{\partial \check{\Psi}}{\partial Y}+{\rm div}   \dfrac{\partial \check{\Psi}}{\partial \nabla Y} &=~0,\\
  \check{\mu}_\rho - \dfrac{ \partial \check{\Psi}}{\partial \rho} + {\rm div}\dfrac{\partial \check{\Psi}}{\partial\nabla \rho} &=~0.
  \end{align}
\end{subequations}

\subsection{Mass measure free energy}\label{subsec: mass measure free energy}
The mass measure free energy is postulated to belong to the class:
\begin{align}\label{eq: appendix: class c psi}
  \psi = \doublecheck{\psi}(\rho,\nabla \rho, Y,\nabla Y),     
\end{align}
and the chemical potential quantities defined as:
\begin{subequations}
\begin{align}
    \doublecheck{\mu}_Y =&~ \dfrac{ \partial \doublecheck{\psi}}{\partial Y} - \frac{1}{\rho}{\rm div}\left(\rho \dfrac{\partial \doublecheck{\psi}}{\partial\nabla Y}\right),\\
    \doublecheck{\mu}_\rho =&~ \dfrac{ \partial \doublecheck{\psi}}{\partial \rho} - \frac{1}{\rho}{\rm div}\left(\rho \dfrac{\partial \doublecheck{\psi}}{\partial\nabla \rho}\right).
\end{align}
\end{subequations}
We proceed with the evaluation of the evolution of the energy \eqref{eq: appendix: class c psi} and find via Reynolds transport theorem and the mass balance equation \eqref{eq: appendix mass 1}:
\begin{align}\label{eq: appendix: psi derivation 1: mass-measure c}
    \dfrac{{\rm d}}{{\rm d}t}\displaystyle\int_{\mathcal{R}(t)} \rho\doublecheck{\psi} ~{\rm d}v  =&~ \displaystyle\int_{\mathcal{R}(t)} \rho \dot{\doublecheck{\psi}} ~{\rm d}v\nn\\
    =&~ \displaystyle\int_{\mathcal{R}(t)} \rho \left(\dfrac{\partial \doublecheck{\psi}}{\partial \rho} \dot{\rho} + \dfrac{\partial \doublecheck{\psi}}{\partial \nabla \rho}\cdot \dot{\overline{\nabla \rho}} +\dfrac{\partial \doublecheck{\psi}}{\partial Y} \dot{Y} + \dfrac{\partial \doublecheck{\psi}}{\partial \nabla Y}\cdot \dot{\overline{\nabla Y}}\right) ~{\rm d}v.
\end{align}
On the account of the relation
\begin{align}\label{eq: appendix: relation grad phi}
    \dot{\overline{\nabla \omega}} = \nabla (\dot{\omega}) - (\nabla \omega)^T\nabla \bv,
\end{align}
for $\omega = \rho, Y$ we arrive at:
\begin{align}
    \dfrac{{\rm d}}{{\rm d}t}\displaystyle\int_{\mathcal{R}(t)} \rho\doublecheck{\psi} ~{\rm d}v  =&~ \displaystyle\int_{\mathcal{R}(t)} \rho \doublecheck{\mu}_Y \dot{Y} +\rho\doublecheck{\mu}_\rho \dot{\rho} - \left(\nabla \rho \otimes \dfrac{\partial \rho\doublecheck{\psi}}{\partial \nabla \rho}+\nabla Y \otimes \dfrac{\partial \rho\doublecheck{\psi}}{\partial \nabla Y}\right): \nabla \bv ~{\rm d}v  \nn\\
    &~+ \displaystyle\int_{\partial \mathcal{R}(t)} \left(\dot{\rho} \rho\dfrac{\partial \doublecheck{\psi}}{\partial \nabla \rho}+\dot{Y} \rho\dfrac{\partial \doublecheck{\psi}}{\partial \nabla Y}\right)\cdot \boldsymbol{\nu} ~{\rm d}a.
\end{align}
By substituting the mass balance laws \eqref{eq: appendix mass 1}-\eqref{eq: appendix mass 2} we find:
\begin{align}\label{eq: appendix: Psi 2: mass-measure}
    \dfrac{{\rm d}}{{\rm d}t}\displaystyle\int_{\mathcal{R}(t)} \rho\doublecheck{\psi} ~{\rm d}v =&~ \displaystyle\int_{\mathcal{R}(t)} \nabla \doublecheck{\mu}_Y \cdot \bH  -\rho^2 \doublecheck{\mu}_\rho {\rm div}\bv   + \doublecheck{\mu}_Y\zeta \nn\\
    &~~~~~~~~- \left(\nabla \rho \otimes \dfrac{\partial \rho\doublecheck{\psi}}{\partial \nabla \rho}+\nabla Y \otimes \dfrac{\partial \rho\doublecheck{\psi}}{\partial \nabla Y}\right): \nabla \bv~{\rm d}v \nn\\
    &~+ \displaystyle\int_{\partial \mathcal{R}(t)}\left(-\doublecheck{\mu}_Y \bH+\dot{\rho} \rho\dfrac{\partial \doublecheck{\psi}}{\partial \nabla \rho}+\dot{ c}\rho \dfrac{\partial \doublecheck{\psi}}{\partial \nabla  c}\right)\cdot \boldsymbol{\nu} ~{\rm d}a.
\end{align}
Following the same procedure as in \cref{sec: 2nd law} we find:
\begin{align}\label{eq: appendix: second law subst 1 c-psi based}
    \dfrac{{\rm d}}{{\rm d}t} \mathscr{E} = &~ \displaystyle\int_{\partial \mathcal{R}(t)}\left(\bv^T\doublecheck{\mathbf{T}}-\doublecheck{\mu}_Y \bH+\dot{\rho} \rho\dfrac{\partial \doublecheck{\psi}}{\partial \nabla \rho}+\dot{Y} \rho\dfrac{\partial \doublecheck{\psi}}{\partial \nabla Y}\right)\cdot \boldsymbol{\nu} ~{\rm d}a \nn\\
    &~- \displaystyle\int_{\mathcal{R}(t)}  \left( \doublecheck{\mathbf{T}} + \rho^2 \doublecheck{\mu}_\rho\mathbf{I}  + \rho \nabla \rho \otimes \dfrac{\partial \doublecheck{\psi}}{\partial \nabla \rho} + \rho \nabla Y \otimes \dfrac{\partial \doublecheck{\psi}}{\partial \nabla Y} \right): \nabla \bv\nn\\
    &~~~~~~~~~~-\nabla \doublecheck{\mu}_Y \cdot \bH -\doublecheck{\mu}_Y\zeta ~{\rm d}v,
\end{align}
where we identify the rate of work and dissipation term:
\begin{subequations}\label{eq: appendix: W, D; c based appendix 2}
\begin{align}
    \mathscr{W} =&~   \displaystyle\int_{\partial \mathcal{R}(t)}\left(\bv^T\doublecheck{\mathbf{T}}-\doublecheck{\mu}_Y\bH+\dot{\rho} \rho\dfrac{\partial \doublecheck{\psi}}{\partial \nabla \rho}+\dot{Y} \rho \dfrac{\partial  \doublecheck{\psi}}{\partial \nabla Y}\right)\cdot \boldsymbol{\nu} ~{\rm d}a,\\
    \mathscr{D} =&~  \displaystyle\int_{\mathcal{R}(t)}  \left( \doublecheck{\mathbf{T}} +\rho^2 \doublecheck{\mu}_\rho\mathbf{I}+ \rho \nabla \rho \otimes \dfrac{\partial \doublecheck{\psi}}{\partial \nabla \rho}+ \rho\nabla Y \otimes  \dfrac{\partial \doublecheck{\psi}}{\partial \nabla Y}  \right): \nabla \bv \nn\\
    &\quad\quad\quad\quad -\nabla \doublecheck{\mu}_Y \cdot \bH  -\doublecheck{\mu}_Y\zeta ~{\rm d}v.
\end{align}
\end{subequations}
From the design condition $\mathscr{D}\geq 0$ we obtain the local modeling restriction:
\begin{align}\label{eq: appendix: second law c 2}
    \left( \doublecheck{\mathbf{T}} +\rho^2 \doublecheck{\mu}_\rho\mathbf{I} + \rho \nabla \rho \otimes \dfrac{\partial \doublecheck{\psi}}{\partial \nabla \rho}+ \nabla Y \otimes \rho\dfrac{\partial \doublecheck{\psi}}{\partial \nabla Y} \right): \nabla \bv -\nabla \doublecheck{\mu}_Y \cdot \bH -\doublecheck{\mu}_Y\zeta\geq 0.
\end{align}
Based on the modeling restriction \eqref{eq: appendix: second law c 2} we introduce with the constitutive classes:
\begin{subequations}\label{eq: appendix class T h: c Psi}
\begin{align}
    \mathbf{T} =&~ \doublecheck{\mathbf{T}}(\nabla \bv, \rho,\nabla \rho,Y, \nabla Y,\doublecheck{\mu}_\rho, \nabla \doublecheck{\mu}_\rho, \doublecheck{\mu}_Y, \nabla \doublecheck{\mu}_Y), \label{eq: appendix: class T appendix2}\\
    \bH =&~ \doublecheck{\bH}(\rho,\nabla \rho,Y, \nabla Y,\doublecheck{\mu}_\rho, \nabla \doublecheck{\mu}_\rho, \doublecheck{\mu}_Y, \nabla \doublecheck{\mu}_Y),\label{eq: appendix: class h appendix2}\\
    \zeta =&~ \doublecheck{\zeta}(\rho, Y,\doublecheck{\mu}_\rho, \doublecheck{\mu}_Y).\label{eq: appendix: class gamma appendix2}
\end{align}
\end{subequations}
We choose closure models that ensure each term in \eqref{eq: appendix: second law 4} remains separately non-negative:
\begin{subequations}\label{eq: appendix: choices 2}
    \begin{align}
    \check{\mathbf{T}} =&~ - \rho \nabla Y \otimes \dfrac{\partial \doublecheck{\psi}}{\partial \nabla Y} - \rho \nabla \rho \otimes \dfrac{\partial \doublecheck{\psi}}{\partial \nabla \rho}  -\rho^2 \doublecheck{\mu}_\rho\mathbf{I}  + \nu (2\mathbf{D}+\lambda({\rm div}\bv) \mathbf{I}),\label{eq: appendix: stress tensor choice 2}\\
    \doublecheck{\bH} =&~ -\doublecheck{\mathbf{M}} \nabla \doublecheck{\mu}_Y,\\
    \doublecheck{\zeta} =&~ -\doublecheck{m} \doublecheck{\mu}_Y,
\end{align}
\end{subequations}
where the mobility tensor $\doublecheck{\mathbf{M}}$ is positive definite, depends on the same variables as \eqref{eq: appendix: class h appendix2}, and satisfies $\doublecheck{\mathbf{M}}|_{Y=0,1}=0$. The scalar mobility $\doublecheck{m} \geq 0$ shares the same dependencies as \eqref{eq: appendix: class gamma appendix2} and satisfies $\doublecheck{m}|_{Y=0,1}=0$.

This concludes the constitutive modeling, with the system given by:
\begin{subequations}\label{eq: appendix: model rho c appendix end}
  \begin{align}
   \partial_t (\rho \bv) + {\rm div} \left( \rho \bv\otimes \bv \right) + \nabla \doublecheck{p} + {\rm div} \left( \nabla \rho \otimes \rho\dfrac{\partial \doublecheck{\psi}}{\partial \nabla \rho} +\nabla Y\otimes \rho\dfrac{\partial \doublecheck{\psi}}{\partial \nabla Y}  \right) & \nn\\
    - {\rm div} \left(   \nu (2\mathbf{D}+\lambda({\rm div}\bv) \mathbf{I}) \right)-\rho\mathbf{g} &=~ 0, \label{eq: appendix: model v with const mod: mom 3}\\
 \partial_t \rho + {\rm div}(\rho \bv) &=~ 0, \label{eq: appendix: model v with const mod: cont 3} \\
  \partial_t (\rho Y) + {\rm div}(\rho Y \bv) - {\rm div} \left(\doublecheck{\mathbf{M}}\nabla \doublecheck{\mu}_Y\right)  + \doublecheck{m} \doublecheck{\mu}_Y &=~0,\label{eq: appendix: model v with const mod: PF 3}\\
  \doublecheck{\mu}_Y -\dfrac{ \partial \doublecheck{\psi}}{\partial Y} + \frac{1}{\rho}{\rm div}\left(\rho \dfrac{\partial \doublecheck{\psi}}{\partial\nabla Y}\right)&=~0,\\
  \doublecheck{\mu}_\rho - \dfrac{ \partial \doublecheck{\psi}}{\partial \rho} + \frac{1}{\rho}{\rm div}\left(\rho \dfrac{\partial \doublecheck{\psi}}{\partial\nabla \rho}\right) &=~0,
  \end{align}
\end{subequations}
with the pressure $\doublecheck{p} = \rho^2 \doublecheck{\mu}_\rho$.

\begin{proposition}[Equivalence binary mixture models]\label[proposition]{prop: Equivalence binary mixture models}
    The mixture models  \eqref{eq: appendix: model v with const mod 2} and \eqref{eq: appendix: model rho c appendix end} are equivalent.
\end{proposition}
\begin{proof}
This follows from identifying $ \check{\Psi} = \rho \doublecheck{\psi} $, $ \check{\mathbf{M}} = \doublecheck{\mathbf{M}} $, and $ \check{m} = \doublecheck{m} $. In particular, the first identification provides $ \check{\mu}_Y = \rho \doublecheck{\mu}_Y $ and the identification of the pressure follows from:
\begin{align}
    \check{p} = \rho \check{\mu}_\rho - \check{\Psi} = \rho (\rho \doublecheck{\mu}_\rho +\doublecheck{\psi}) - \rho \doublecheck{\psi} = \rho^2 \doublecheck{\mu}_\rho= \doublecheck{p},
\end{align}
with $\check{\mu}_\rho=\rho \doublecheck{\mu}_\rho +\doublecheck{\psi}$.%This confirms the equivalence of the mixture models \eqref{eq: appendix: model v with const mod} and \eqref{eq: appendix: model rho c appendix end}.    
\end{proof}
Finally, utilizing \cref{lem: binary fluids Korteweg identity mass} the model may be written in the compact form:
\begin{subequations}
  \begin{align}
  \partial_t (\rho \bv) + {\rm div} \left( \rho \bv\otimes \bv \right) + \rho \nabla (\rho \doublecheck{\mu}_\rho + \doublecheck{\psi}) - \rho \doublecheck{\mu}_Y \nabla Y&\nn\\
    - {\rm div} \left(   \nu (2\mathbf{D}+\lambda({\rm div}\bv) \mathbf{I}) \right)-\rho\mathbf{g} &=~ 0, \\
  \partial_t \rho + {\rm div}(\rho \bv) &=~ 0,\\
 \partial_t (\rho Y) + {\rm div}(\rho Y \bv) - {\rm div} \left(\doublecheck{M}\nabla\doublecheck{\mu}_Y\right)  + \doublecheck{m} \doublecheck{\mu}_Y &=~0.
 \end{align}
\end{subequations} 

\section{Hyperbolicity in $d$ dimensions}\label{app: hyperbolic}

We write the system in the matrix-vector form:
\begin{align}
  \partial_t \mathbf{W} + \mathbf{A}_{x_1} \partial_{x_1} \mathbf{W} + \mathbf{A}_{x_2} \partial_{x_2} \mathbf{W}  + \mathbf{A}_{x_3} \partial_{x_3} \mathbf{W}   = 0,
\end{align}
where
\begin{align}
  \mathbf{W} = \begin{bmatrix}
      p_1 \\[6pt]
      \vdots \\[6pt]
      p_N\\[6pt]
      v_{x_1} \\[6pt]
      v_{x_2} \\[6pt]
      v_{x_3}
  \end{bmatrix}, \quad &\mathbf{A}_{x_1} =\begin{bmatrix}
    v_{x_1}           & 0           & \cdots  &        & 0           & \sum_\mB \trho_\mB a_{1\mB}^2 & 0 & 0\\[6pt]
    0           & v_{x_1}           & 0      & \cdots  & 0           & \sum_\mB \trho_\mB a_{2\mB}^2& 0 & 0\\[6pt]
    \vdots      &  \ddots           & \ddots & \ddots       &  \vdots           & \vdots& 0 & 0\\[6pt]
    0           &   \cdots    &  0      & v_{x_1}      &      0      & \sum_\mB \trho_\mB a_{(N-1)\mB}^2& 0 & 0\\[6pt]
    0           &   \cdots    &        & 0      &      v_{x_1}      & \sum_\mB \trho_\mB a_{N\mB}^2& 0 & 0\\[6pt]
    \rho^{-1} &      & \cdots  &        & \rho^{-1}   & v_{x_1} & 0 & 0 \\[6pt]
    0 & & & &\cdots & 0& v_{x_1} & 0 \\[6pt]
    0 & & & &\cdots & & 0& v_{x_1} 
  \end{bmatrix},\nn\\[6pt]
  &\mathbf{A}_{x_2} =\begin{bmatrix}
    v_{x_2}           & 0           & \cdots  &        & 0           & 0 & \sum_\mB \trho_\mB a_{1\mB}^2 & 0\\[6pt]
    0           & v_{x_2}           & 0      & \cdots  & 0           & 0 & \sum_\mB \trho_\mB a_{2\mB}^2 & 0\\[6pt]
    \vdots      &  \ddots           & \ddots & \ddots       &  \vdots         & 0  & \vdots & 0\\[6pt]
    0           &   \cdots    &  0      & v_{x_2}      &      0      & 0 & \sum_\mB \trho_\mB a_{(N-1)\mB}^2 & 0 \\[6pt]
    0           &   \cdots    &        & 0      &      v_{x_2}      & 0 & \sum_\mB \trho_\mB a_{N\mB}^2 & 0\\[6pt]
    0 & & \cdots& & 0 &v_{x_2} &0 & 0 \\[6pt]
    \rho^{-1} &      & \cdots  &        & \rho^{-1}  & 0 & v_{x_2} & 0   \\[6pt]
    0 & & & &\cdots & & 0& v_{x_2}
  \end{bmatrix},\nonumber \\[6pt] &\mathbf{A}_{x_3} =\begin{bmatrix}
    v_{x_3}           & 0           & \cdots  &        & 0           & 0 & 0 & \sum_\mB \trho_\mB a_{1\mB}^2\\[6pt]
    0           & v_{x_3}           & 0      & \cdots  & 0           & 0 & 0 & \sum_\mB \trho_\mB a_{2\mB}^2\\[6pt]
    \vdots      &  \ddots           & \ddots & \ddots       &  \vdots         & 0 & 0   & \vdots \\[6pt]
    0           &   \cdots    &  0      & v_{x_3}      &      0      & 0 & 0 & \sum_\mB \trho_\mB a_{(N-1)\mB}^2\\[6pt]
    0           &   \cdots    &        & 0      &      v_{x_3}      & 0 & 0 & \sum_\mB \trho_\mB a_{N\mB}^2\\[6pt]
    0 & & \cdots& &0 &v_{x_3} &0 & 0 \\[6pt]
    0 & & &\cdots & &0 &v_{x_3} & 0 \\[6pt]
    \rho^{-1} &      & \cdots  &        & \rho^{-1} & 0 & 0   & v_{x_3}  
  \end{bmatrix}.
\end{align}
The matrices $\mathbf{A}_{x_i}, i=1,2,3$ have the $N+d$ eigenvalues:
    \begin{align}
    \lambda_{i,1} = v_i - a < \lambda_{i,2} = ... = \lambda_{i,N+d-1} = v_i < \lambda_{i,N+d} = v_i+a.
\end{align}
The corresponding right eigenvectors are:{\small
\begin{subequations}\begin{align}
    \br_{x_1,1} = \begin{bmatrix}
        - \sum_\mB \trho_\mB a_{1\mB}^2\\
        \vdots \\
        - \sum_\mB \trho_\mB a_{N\mB}^2\\
        a\\0\\0
    \end{bmatrix}, ~
    \br_{x_1,2} = \begin{bmatrix}
        - 1\\
        1 \\
        0\\
        \vdots \\
        \\
        0
    \end{bmatrix}, ~
    \br_{x_1,3} = \begin{bmatrix}
        - 1\\
        0 \\
        1\\
        0\\
        \vdots \\
        0
    \end{bmatrix}, \cdots ,~
    \br_{x_1,N} = \begin{bmatrix}
        - 1\\
        0 \\
        \vdots \\
        0\\
        1\\
        0\\0\\0
    \end{bmatrix}, \nn\\
    \br_{x_1,N+1} = \begin{bmatrix}
        0\\
        \vdots \\
        0\\
        1\\
        0
    \end{bmatrix},~
    \br_{x_1,N+2} = \begin{bmatrix}
        0 \\
        \vdots \\
        0\\
        0\\
        1
    \end{bmatrix}, ~
    \br_{x_1,N+3} =\begin{bmatrix}
        \sum_\mB \trho_\mB a_{1\mB}^2\\
        \vdots \\
        \sum_\mB \trho_\mB a_{N\mB}^2\\
        a\\0\\0
    \end{bmatrix},\\
        \br_{x_2,1} = \begin{bmatrix}
        - \sum_\mB \trho_\mB a_{1\mB}^2\\
        \vdots \\
        - \sum_\mB \trho_\mB a_{N\mB}^2\\
        0\\a\\0
    \end{bmatrix}, ~
    \br_{x_2,2} = \begin{bmatrix}
        - 1\\
        1 \\
        0\\
        \vdots \\
        \\
        0
    \end{bmatrix}, ~
    \br_{x_2,3} = \begin{bmatrix}
        - 1\\
        0 \\
        1\\
        0\\
        \vdots \\
        0
    \end{bmatrix}, \cdots ,~
    \br_{x_2,N} = \begin{bmatrix}
        - 1\\
        0 \\
        \vdots \\
        0\\
        1\\
        0\\0\\0
    \end{bmatrix}, \nn\\
    \br_{x_2,N+1} = \begin{bmatrix}
        0\\
        \vdots \\
        1\\
        0\\
        0
    \end{bmatrix},~
    \br_{x_2,N+2} = \begin{bmatrix}
        0 \\
        \vdots \\
        0\\
        0\\
        1
    \end{bmatrix}
    , ~
    \br_{x_2,N+3} =\begin{bmatrix}
        \sum_\mB \trho_\mB a_{1\mB}^2\\
        \vdots \\
        \sum_\mB \trho_\mB a_{N\mB}^2\\
        0\\a\\0
    \end{bmatrix},\\
        \br_{x_3,1} = \begin{bmatrix}
        - \sum_\mB \trho_\mB a_{1\mB}^2\\
        \vdots \\
        - \sum_\mB \trho_\mB a_{N\mB}^2\\
        0\\0\\a
    \end{bmatrix}, ~
    \br_{x_3,2} = \begin{bmatrix}
        - 1\\
        1 \\
        0\\
        \vdots \\
        \\
        0
    \end{bmatrix}, ~
    \br_{x_3,3} = \begin{bmatrix}
        - 1\\
        0 \\
        1\\
        0\\
        \vdots \\
        0
    \end{bmatrix}, \cdots ,~
    \br_{x_3,N} = \begin{bmatrix}
        - 1\\
        0 \\
        \vdots \\
        0\\
        1\\
        0\\0\\0
    \end{bmatrix}, \nn\\
    \br_{x_3,N+1} = \begin{bmatrix}
        0\\
        \vdots \\
        1\\
        0\\
        0
    \end{bmatrix},~
    \br_{x_3,N+2} = \begin{bmatrix}
        0 \\
        \vdots \\
        0\\
        1\\
        0
    \end{bmatrix}, ~
    \br_{x_3,N+3} =\begin{bmatrix}
        \sum_\mB \trho_\mB a_{1\mB}^2\\
        \vdots \\
        \sum_\mB \trho_\mB a_{N\mB}^2\\
        0\\0\\a
    \end{bmatrix},
\end{align}
\end{subequations}}
and the left eigenvectors are:\allowdisplaybreaks
{\small
\begin{subequations}\begin{align}
    \bl_{x_1,1} = \begin{bmatrix}
        - \frac{1}{2 \rho a^2}\\
        \vdots \\
        - \frac{1}{2 \rho a^2}\\
         \frac{1}{2 a}\\
         0\\0
    \end{bmatrix}, ~
    \bl_{x_1,2} = \begin{bmatrix}
        - \frac{\sum_\mB \trho_\mB a_{2\mB}^2}{\sum_\mB \trho_\mB (a_{1\mB}^2+a_{2\mB}^2)}\\
          \frac{\sum_\mB \trho_\mB a_{1\mB}^2}{\sum_\mB \trho_\mB (a_{1\mB}^2+a_{2\mB}^2)}\\
        0\\
        \vdots\\
        0
    \end{bmatrix}, ~
    \bl_{x_1,3} = \begin{bmatrix}
        - \frac{\sum_\mB \trho_\mB a_{3\mB}^2}{\sum_\mB \trho_\mB (a_{1\mB}^2+a_{3\mB}^2)}\\
        0\\
          \frac{\sum_\mB \trho_\mB a_{1\mB}^2}{\sum_\mB \trho_\mB (a_{1\mB}^2+a_{3\mB}^2)}\\
        0\\
        \vdots\\
        0
    \end{bmatrix}, \cdots ,\nn \\
    \bl_{x_1,N} = \begin{bmatrix}
        - \frac{\sum_\mB \trho_\mB a_{N\mB}^2}{\sum_\mB \trho_\mB (a_{1\mB}^2+a_{N\mB}^2)}\\
        0\\
        \vdots \\
        0\\
          \frac{\sum_\mB \trho_\mB a_{1\mB}^2}{\sum_\mB \trho_\mB (a_{1\mB}^2+a_{N\mB}^2)}\\
        0\\
         0\\0
    \end{bmatrix},~
    \bl_{x_1,N+1} = \begin{bmatrix}
        0\\
        \vdots \\
        0\\
        1\\
        0
    \end{bmatrix},~
    \bl_{x_1,N+2} = \begin{bmatrix}
        0 \\
        \vdots \\
        0\\
        0\\
        1
    \end{bmatrix}, ~
    \bl_{x_1,N+3} =\begin{bmatrix}
         \frac{1}{2 \rho a^2}\\
        \vdots \\
         \frac{1}{2 \rho a^2}\\
         \frac{1}{2 a}\\
         0\\0
    \end{bmatrix},\\
        \bl_{x_2,1} = \begin{bmatrix}
        - \frac{1}{2 \rho a^2}\\
        \vdots \\
        - \frac{1}{2 \rho a^2}\\
         0\\ \frac{1}{2 a}\\0
    \end{bmatrix}, ~
    \bl_{x_2,2} = \begin{bmatrix}
        - \frac{\sum_\mB \trho_\mB a_{2\mB}^2}{\sum_\mB \trho_\mB (a_{1\mB}^2+a_{2\mB}^2)}\\
          \frac{\sum_\mB \trho_\mB a_{1\mB}^2}{\sum_\mB \trho_\mB (a_{1\mB}^2+a_{2\mB}^2)}\\
        0\\
        \vdots\\
        0
    \end{bmatrix}, ~
    \bl_{x_2,3} = \begin{bmatrix}
        - \frac{\sum_\mB \trho_\mB a_{3\mB}^2}{\sum_\mB \trho_\mB (a_{1\mB}^2+a_{3\mB}^2)}\\
        0\\
          \frac{\sum_\mB \trho_\mB a_{1\mB}^2}{\sum_\mB \trho_\mB (a_{1\mB}^2+a_{3\mB}^2)}\\
        0\\
        \vdots\\
        0
    \end{bmatrix}, \cdots ,\nn\\
    \bl_{x_2,N} = \begin{bmatrix}
        - \frac{\sum_\mB \trho_\mB a_{N\mB}^2}{\sum_\mB \trho_\mB (a_{1\mB}^2+a_{N\mB}^2)}\\
        0\\
        \vdots \\
        0\\
          \frac{\sum_\mB \trho_\mB a_{1\mB}^2}{\sum_\mB \trho_\mB (a_{1\mB}^2+a_{N\mB}^2)}\\
        0\\
         0\\0
    \end{bmatrix},~
    \bl_{x_2,N+1} = \begin{bmatrix}
        0\\
        \vdots \\
        1\\
        0\\
        0
    \end{bmatrix},~
    \bl_{x_2,N+2} = \begin{bmatrix}
        0 \\
        \vdots \\
        0\\
        0\\
        1
    \end{bmatrix}, ~
    \bl_{x_2,N+3} =\begin{bmatrix}
         \frac{1}{2 \rho a^2}\\
        \vdots \\
         \frac{1}{2 \rho a^2}\\
         0\\ \frac{1}{2 a}\\0
    \end{bmatrix},\\
        \bl_{x_3,1} = \begin{bmatrix}
        - \frac{1}{2 \rho a^2}\\
        \vdots \\
        - \frac{1}{2 \rho a^2}\\
         0\\0\\\frac{1}{2 a}
    \end{bmatrix}, ~
    \bl_{x_3,2} = \begin{bmatrix}
        - \frac{\sum_\mB \trho_\mB a_{2\mB}^2}{\sum_\mB \trho_\mB (a_{1\mB}^2+a_{2\mB}^2)}\\
          \frac{\sum_\mB \trho_\mB a_{1\mB}^2}{\sum_\mB \trho_\mB (a_{1\mB}^2+a_{2\mB}^2)}\\
        0\\
        \vdots\\
        0
    \end{bmatrix}, ~
    \bl_{x_3,3} = \begin{bmatrix}
        - \frac{\sum_\mB \trho_\mB a_{3\mB}^2}{\sum_\mB \trho_\mB (a_{1\mB}^2+a_{3\mB}^2)}\\
        0\\
          \frac{\sum_\mB \trho_\mB a_{1\mB}^2}{\sum_\mB \trho_\mB (a_{1\mB}^2+a_{3\mB}^2)}\\
        0\\
        \vdots\\
        0
    \end{bmatrix}, \cdots ,\nn\\
    \bl_{x_3,N} = \begin{bmatrix}
        - \frac{\sum_\mB \trho_\mB a_{N\mB}^2}{\sum_\mB \trho_\mB (a_{1\mB}^2+a_{N\mB}^2)}\\
        0\\
        \vdots \\
        0\\
          \frac{\sum_\mB \trho_\mB a_{1\mB}^2}{\sum_\mB \trho_\mB (a_{1\mB}^2+a_{N\mB}^2)}\\
        0\\
         0\\0
    \end{bmatrix},~
    \bl_{x_3,N+1} = \begin{bmatrix}
        0\\
        \vdots \\
        1\\
        0\\
        0
    \end{bmatrix},~
    \bl_{x_3,N+2} = \begin{bmatrix}
        0 \\
        \vdots \\
        0\\
        1\\
        0
    \end{bmatrix}, ~
    \bl_{x_3,N+3} =\begin{bmatrix}
         \frac{1}{2 \rho a^2}\\
        \vdots \\
         \frac{1}{2 \rho a^2}\\
         0\\0\\\frac{1}{2 a}
    \end{bmatrix}. 
\end{align}
\end{subequations}}
The eigenvectors are scaled so that $\br_j\cdot\bl_j=1$ for all $j=1,...,N+1$. Both right and left eigenvectors are linearly independent and span $\mathbb{R}^{N+d}$. This implies that the system is  hyperbolicity if the eigenvalues are real-valued, which holds when $a_{\mA\mB}^2={\rm d}p_\mA/{\rm d}\trho_\mB >0$.  Similar to the one-dimensional case, fields associated with $\lambda_2,...,\lambda_{N+{2}}$ correspond to contact discontinuities, and those associated with $\lambda_1,\lambda_{N+{3}}$ correspond to shock or rarefaction waves, respectively.

\bibliographystyle{unsrtnat}
\bibliography{references}

\end{document}